\documentclass[onecolumn,draftclsnofoot,12pt]{IEEEtran}
\usepackage{tipa}

\usepackage{setspace}
\doublespace
\usepackage{cite}
\usepackage{bbm}

\usepackage{graphicx,bm,hyperref}
\usepackage{multicol}
\usepackage{amsmath}
\usepackage{amsfonts}
\usepackage{amssymb}
\usepackage{graphicx}
\makeatletter

\newcommand{\Rmnum}[1]{\expandafter\@slowromancap\romannumeral #1@}
\makeatother

\newtheorem{theorem}{Theorem}
\newtheorem{Lemma}[theorem]{Lemma}

\newenvironment{proof}[1][Proof]{\begin{trivlist}
\item[\hskip \labelsep {\bfseries #1}]}{\end{trivlist}}

\newcommand{\qed}{\nobreak \ifvmode \relax \else
      \ifdim\lastskip<1.5em \hskip-\lastskip
      \hskip1.5em plus0em minus0.5em \fi \nobreak
      \vrule height0.75em width0.5em depth0.25em\fi}

\begin{document}

\title{Design of Binary Network Codes for Multi-user Multi-way Relay Networks}

\author{\IEEEauthorblockN {Ang Yang, Zesong Fei, Chengwen Xing, Ming Xiao, Jinhong Yuan, and Jingming
Kuang}}


\maketitle

\begin{abstract}
We study multi-user multi-way relay networks where $N$ user nodes exchange their information through a
single relay node. We use network coding in the relay to increase the throughput. Due to the
limitation of complexity, we only consider the binary multi-user network coding (BMNC) in the relay.
We study BMNC matrix (in GF(2)) and propose several design criteria on the BMNC matrix to improve the
symbol error probability (SEP) performance. Closed-form expressions of the SEP of the system are
provided. Moreover, an upper bound of the SEP is also proposed to provide further insights on system
performance. Then BMNC matrices are designed to minimize the error probabilities.
\end{abstract}


\let\oldthefootnote\thefootnote
\renewcommand\thefootnote{}
\footnote{Ang Yang, Zesong Fei, Chengwen Xing and Jingming Kuang are with School of Information and
Electronics, Beijing Institute of Technology, Beijing, China (e-mail: \{cool\_yang, feizesong,
chengwenxing, jmkuang\}@bit.edu.cn).
\par Ming Xiao is with the ACCESS Linnaeus Center, Royal Institute of Technology, Stockholm, Sweden (e-mail:
ming.xiao@ee.kth.se).
\par Jinhong Yuan is with the School of Electrical Engineering and
Telecommunications, University of New South Wales, Australia (e-mail: j.yuan@unsw.edu.au).}
\let\thefootnote\oldthefootnote

\IEEEpeerreviewmaketitle

\section{Introduction}
Network coding (NC) is considered as a potentially powerful tool for efficient information
transmission in wireless networks, where data flows coming from multiple sources or to different sinks
are combined to increase throughput, reduce delay, or enhance robustness
\cite{MUNC:NC_flow,MUNC:NC_flow_2,Naka_df:Ning_Cai_Secure_NC}. Consider a two-way wireless system
where two source nodes communicate with each other through the aid of one relay node
\cite{Naka_df:TaoCui_Dif_modula_NC_two-way,Naka_df:TaoCui_memoryless_two-way,MUNC:feifei_gao_two_way_1,MUNC:y_jing_two_way,1687-1499-2012-66}.
With network coding, each of the two transceivers employs one time slot to transmit a packet to the
relay in a conventional time-division multiple access (TDMA) scheme. Next, the relay takes the
exclusive-or of these two packets and broadcasts the result during the third time slot. Armed with the
packet it sent to the relay, each of the transceivers can then recover the data originating at the
other relay, with the network having only used three slots rather than the traditional four. In what
follows, we shall extend NC to the multi-user, multi-hop, multi-relay and multi-radio wireless ad hoc
networks, which are introduced in
\cite{Naka_df:Peng_Liu_AF_Limited,Naka_df:DSTC_laneman,Jing:Distributed_STC_in_WSN,rvq:david,rvq:Eric}.

As previous related work, the sink bit error probability (BEP) for the coded network with memoryless
and independent channels is investigated in \cite{MUNC:Ming_Xiao_1}. The alphabet size of the code is
GF$(2^m)$. In \cite{MUNC:Ming_Xiao_2}, finite-field network coding (FFNC) is designed for
multiple-user multiple-relay (MUMR) wireless networks with quasi-static fading channels. For high rate
regions, FFNC has significantly better performance than superposition coding. In
\cite{MUNC:Multiuser_Two-Way_Relaying_Detection_and_Interference_Management}, using code division
multiple access (CDMA) of an interference limited system, a jointly demodulate-and-XOR forward
(JD-XOR-F) relaying scheme is proposed, where all users transmit to the relay simultaneously followed
by the relay broadcasting an estimate of the XORed symbol for each user pair. The problem of joint
resource allocation for OFDMA assisted two-way relay system is studied in
\cite{MUNC:feifei_gao_two_way_multi_user} and the objective function is to maximize the sum-rate
through joint subcarrier allocation, sub-carrier pairing, and power allocation, under the individual
power constraints at each transmitting node. Several beamforming schemes are proposed in
\cite{MUNC:Analogue_Network_Coding_for_Multi-Pair} for the scenario where multiple pairs of users
exchange information within pair, with the help of a dedicated multi-antenna relay. A cooperation
protocol based on complex-field wireless network coding is developed in a network with $N$ sources and
one destination \cite{MUNC:High-Throughput_Multi-Source_Cooperation_via_Complex-Field_Network_Coding}.
To deal with decoding errors at sources, selective- and adaptive-forwarding protocols are also
developed at no loss of diversity gain. For the multiple-access relay network, the capacity
approaching behavior of the joint network LDPC code is analyzed in \cite{MUNC:yingli1,MUNC:yingli2}.

Above literatures focus on the information exchange of multiple pairs of users with or without the
assistance of one relay node. For more general cases, in a practical network, there are multiple
relays or multiple hops. In \cite{MUNC:Yuan_Xiao_NC_1}, new approaches to LDPC code design for a
multi-source single-relay FDMA system are explored, under the assumption of uniform phase-fading
Gaussian channels. In \cite{MUNC:Yuan_Xiao_NC_2}, a binary field NC design over a multiple-source
multiple-relay wireless network over slow-fading channels is studied. In
\cite{MUNC:A_Novel_Collaboration_Scheme_for_Multi-Channel/Interface_NC}, a novel scheme of
multi-channel/interface network coding is proposed, which is based on the combination of a new concept
of coded-overhearing and coding-aware channel assignment. In \cite{MUNC:Shi_jin_1}, the power
allocation policies are investigated across the relays for automatic gain control (AGC)-based
amplify-and-forward (AF) distributed space-time code (DSTC) systems in the two-way relay networks. In
\cite{MUNC:Cross_Layer_Opti_multihop_Parwise_Intersession}, with a new flow-based characterization of
pairwise intersession network coding, an optimal joint coding, scheduling, and rate-control scheme can
be devised and implemented using only the binary XOR operation. In \cite{MUNC:Ray_Liu_STNC}, a novel
concept of wireless network cocast (WNC) \cite{MUNC:Ray_Liu_linear_NC} is considered and its
associated space-time network codes (STNCs) are proposed to achieve the foretold objectives. However,
CDMA-like, FDMA-like and TDMA-like techniques are proposed in \cite{MUNC:Ray_Liu_STNC},
\cite{MUNC:Ray_Liu_linear_NC}, where each symbol is assigned a complex-valued signature waveform, the
dedicated carrier and the symbol duration. In \cite{MUNC:coding}, several interesting properties of
network coding matrices are discussed in a network where $N$ users have independent information to
send to a common base station.

In \cite{MUNC:Xi_Zhang}, it has been shown that for the $N$-way single-channel relay network, it takes
at least $(2N - 1)$ time slots for the linear NC scheme without opportunistic listening to perform a
round of the $N$-way relay, where there are $N$, with $N \ge 2$, end nodes exchanging their
information with the assistant of one $N$-way relay with single antenna. However, \cite{MUNC:Xi_Zhang}
focus on a general linear programming framework for solving the throughput optimization problems and a
joint link scheduling, channel assignment, and routing algorithm for the wireless NC schemes to
closely approximate the optimal solutions. The detailed linear NC for $N$-way single-channel relay
networks, such as how $N$ information packets are encoded into $N-1$ pronumerals, is not investigated
in \cite{MUNC:Xi_Zhang}.

In this paper, we take a step further to investigate the efficient
linear NC for $N$-way single-channel relay network, which is also
discussed in \cite{MUNC:Xi_Zhang}. As shown in
\cite{MUNC:Yuan_Xiao_NC_2}, in the case that the network size (i.e.,
the number of the sources and the number of the relays) and the
frame length (i.e., the number of symbols or Galois filed elements
in a frame) are large, we need to choose a large size of the Galois
fields. Therefore, the encoding complexity of the GF(q) codes will
significantly increase. Since binary network coding is of low
complexity, binary multi-user network coding (BMNC) is considered
here to increase throughput. Several design criteria that the BMNC
matrix should follow to increase the system performance are
provided. Moreover, the effects of the noise and the BMNC matrix are
studied, based on which the symbol error probability (SER) of the
system is provided. To improve the system performance further, BMNC
matrices are designed for arbitrary number of users, which minimize
the bound of SEP.


The paper is organized as follows. In Section II, the system model is introduced. In Section III, BMNC
decoding process and a design criterion on the BMNC matrix are presented. Performance analysis is
shown in Section IV, which includes BMNC matrix analysis, closed-form expressions of the SEP and
throughput of the system, the tight upper bound of SEP of the system and the designed BMNC encoding
matrix. In Section V, the optimality of the matrix given in Section IV-D is discussed. Simulation
results are presented in Section VI and the conclusions are given in Section VII.

Some notations are listed as follows. Symbol $\left( {\bf{A}} \right)_{i,j}$ presents the
$(i,j)^{\rm{th}}$ element of matrix ${\bf{A}}$. Symbols $\oplus$, $\sum {^ \oplus }  $ denote the
addition and the summation in GF(2), respectively. Symbol $\circ$ presents element-wise product and $
\left( {{\bf{A}} \circ {\bf{B}}} \right)_{i,j}  = \left( {\bf{A}} \right)_{i,j} \left( {\bf{B}}
\right)_{i,j} $. Symbol $ \prod\limits_{}^{} {^ \circ  } $ presents the element-wise product of
multiple matrices or vectors.

\section{System Model}
Consider a wireless network with $N$ user nodes $U_i$, $i=1,\ldots, N$, and one relay node $R$, as
shown in Fig. 1. Each node has only one antenna, which can be used for both transmission and
reception. For practical services such as video conference in which each user may want to have a
discussion with the other users, the $N$ users need the information of other users and they exchange
information with the assistance of the relay $R$. Without loss of generality, in one time slot, the
exchanged information bit of $U_i$ can be denoted by $x_i$, $1\leq{}i \leq{}N$. Whereas, in practice,
the user nodes and the relay will transmit the information in packets that contains a large number of
symbols. The user nodes will collect all the transmitted packets and then jointly detect them. We
assume that the direct links between the users are not available. All communications must be through
the relay.

Take $U_i$ for example, it needs the information from the other $N-1$ users, while the other $N-1$
users need the information of $U_i$. In the traditional scheme, considering the time division
transmission schemes, the traditional scheme needs $2N$ time slots to finish the information exchange,
where $N$ time slots are used for the relay to receive the $N$ information bits of the $N$ user nodes
and the other $N$ time slots are used for the user nodes to receive the $N$ information bits from the
relay.

In order to improve the system performance, we propose a BMNC scheme, in which only $2N-1$ time slots
are used. In this scheme, the transmission can be divided into two consecutive phases. 1) In the
source transmission phase, each user node sends its own information to the relay node, which takes $N$
time slots. The relay receives and then detects the $N$ information bits from the $N$ users. 2) In the
relay transmission phase, the relay linearly combines the detected information bits, and then
broadcasts the combined information bits to all the users. Since each user knows its own information,
only the information bits of other $N-1$ users are needed. Thus, at least $N-1$ information bits
should be broadcasted from the relay to all the users. Finally, the BMNC scheme takes $2N-1$ time
slots to achieve the information exchange.

In the source transmission phase, the received symbols at the relay node are
\begin{equation}
\label{the matrix of the coding in relay 1}
\begin{split}
{\bf{y}}_0 = {{\bf{H}}_0}\textltailm({\bf{x}}) + {{\bf{N}}_0} ,
\end{split}
\end{equation}where ${\bf{y}}_0 = {\left[ {\begin{array}{*{20}{c}}
{{y_{0,1}}}&{{y_{0,2}}}& \ldots &{{y_{0,N}}} \end{array}} \right]^T}$ denotes the received signals at
the relay, \({{\bf{H}}_0} = diag \{ \\ {{h _{0,1}} , {h _{0,2}}, \ldots , {h _{0,N}}} \} \) denotes
the fading coefficients, ${\bf{x}} = {\left[ {\begin{array}{*{20}{c}} {{x_1}}&{{x_2}}& \ldots &{{x_N}}
\end{array}} \right]^T}$ denotes the bits of the users, $\textltailm(.)$ denotes the modulation
transformation and $\textltailm({\bf{x}})$ denotes the transmitted symbols of the users, and
${\bf{N}}_0 = {\left[ {\begin{array}{*{20}{c}} {{n_{0,1}}}&{{n_{0,2}}}& \ldots &{{n_{0,N}}}
\end{array}} \right]^T}$ denotes the additive white Gaussian noise (AWGN) with zero mean.

Then the relay detects the received information and obtains an estimation of the source bits
${\bf{\tilde x}} = {\left[ {\begin{array}{*{20}{c}} {{\tilde x_1}}&{{\tilde x_2}}& \ldots &{{\tilde
x_N}}
\end{array}} \right]^T}$. Then linearly network coding is proposed to combine $N$
information bits into $N-1$ bits, which can be shown as follows
\begin{equation}
\label{the matrix of the coding in relay 2}
\begin{split}
\left( {{\bf{F\tilde x}}} \right)\bmod \left( 2 \right)= {\bf{r}},
\end{split}
\end{equation}where \( {\bf{F}} \) is
the network encoding matrix in GF(2), vector \( {\bf{r}} = \left[ {\begin{array}{*{20}c} {r_1 } & {r_2
} &{\cdots}  & {r_{N - 1} }
\end{array}} \right]^T   \) denotes the $N-1$ information bits that
the relay will broadcast. Note that $r_i$ is the information bit to be transmitted in time slot $i$,
$i \in [1,N-1]$. The network encoding matrix \( {\bf{F}} \) can be described as
\begin{equation}
\label{the matrix of H}
\begin{split}
 {\bf{F}} & = \left[ {\begin{array}{*{20}c}
   {f_{1 ,1 } } &  {f_{1 ,2 } } & \cdots  & {f_{1 ,N } }  \\
   {f_{2 ,1 } } &  {f_{2 ,2 } } & \cdots  & {f_{2 ,N } }  \\
    \vdots  &            \vdots         & \ddots  &  \vdots   \\
   {f_{N - 1 ,1 } } & {f_{N - 1 ,2 } } & \cdots  & {f_{N - 1 ,N } }  \\
\end{array}} \right]_{(N-1) \times N} \\
  & =  \left[{\begin{array}{*{20}c}
   {{\bf{f}}_1 }  & {{\bf{f}}_2 } & \cdots  & {{\bf{f}}_N }  \\
\end{array}} \right]_{(N-1) \times N} ,
\end{split}
\end{equation}where \(f_{j ,i } \) is one element of the network encoding
matrix in GF(2) for \( j \in \left[ {1,N - 1} \right], i \in \left[ {1,N} \right] \), which is related
to $U_i$ and the $j$th symbol that the relay sends. Vector \( {\bf{f}}_i  = \left[
{\begin{array}{*{20}c}
   {f_{1 ,i } }  & {f_{2 ,i } }  & \cdots  & {f_{{N - 1} ,i } }  \\
\end{array}} \right]^T
\) is the $i$th column vector of \( {\bf{F}} \), which denotes the relationship between $U_i$ and the
$N-1$ symbols that the relay sends.

The symbol GF(2) is referred to the Galois filed of two elements \cite{lidl1997finite}. In our work,
it consists of 0 and 1. Over GF(2), many well-known but important properties of traditional number
systems, such as real number, rational number etc., are retained: addition has an identity element and
an inverse for every element; multiplication has an identity element ``1'' and an inverse for every
element but ``0''; addition and multiplication are commutative and associative; multiplication is
distributive over addition \cite{lidl1997finite}.

In the relay transmission phase, the information that $U_i$ receives is
\begin{equation}
\label{the matrix of the coding in relay 3}
\begin{split}
{\bf{y}}_i = {{\bf{H }}_i}\textltailm({\bf{r}})+ {{\bf{N}}_i} ,
\end{split}
\end{equation}where ${\bf{y}}_i = {\left[ {\begin{array}{*{20}{c}} {{y_{i,1}}}&{{y_{i,2}}}& \ldots
&{{y_{i,N-1}}}
\end{array}} \right]^T}$ denotes the signals that $U_i$ receives, ${{\bf{H }}_i} = diag\{
\\{{h _{i,1}}, {h _{i,2}},   \ldots ,{h _{i,N-1}}} \}$ denotes the fading coefficients, and
${\bf{N}}_i = {\left[ {\begin{array}{*{20}{c}} {{n_{i,1}}}&{{n_{i,2}}}& \ldots &{{n_{i,N-1}}}
\end{array}} \right]^T}$ denotes the AWGN with zero mean. Then the relay detects
the received information ${\bf{y}}_i $ and obtains ${\bf{\tilde r}}_i = {\left[
{\begin{array}{*{20}{c}} {{\tilde r_{i,1}}}&{{\tilde r_{i,2}}}& \ldots &{{\tilde r_{i,N-1}}}
\end{array}} \right]^T}$. Finally $U_i$ decodes the information of other users through ${\bf{ \tilde
r}}_i, {\bf{F}}, x_i$.

\section{BMNC Decoding Process}
As discussed above, the relay needs to broadcast at least $N-1$ coded bits. However, arbitrary
encoding may cause some users can not decode the source information bits even though there is no noise
in the system. Thus, the network coding process should be designed carefully.

Clearly, $U_i$ only knows its own information $x_i$, the $N-1$ bits ${\bf{ \tilde r}}_i$ that it
detects from the received information and the network coding matrix $\bf{F}$. Then we investigate the
relationship between $x_i$, ${\bf{ \tilde r}}_i$ and $\bf{F}$. From (\ref{the matrix of the coding in
relay 2}), we have
\begin{equation}
\label{decode 1_1}
\begin{split}
 {\bf{r}}
  = \left[ {\begin{array}{*{20}c}
   {\sum\limits_{k = 1}^N {^ \oplus  f_{1 ,k } \tilde x_{k} } }  \\
   {\sum\limits_{k = 1}^N {^ \oplus  f_{2 ,k } \tilde x_{k} } }  \\
    \vdots   \\
   {\sum\limits_{k = 1}^N {^ \oplus  f_{N - 1 ,k } \tilde x_{k} } }  \\
\end{array}} \right]_{(N-1) \times 1}.
\end{split}
\end{equation}

Separating the information of $U_i$ and other users, the above equation can be rewritten as
\begin{equation}
\label{decode 1_1_2}
\begin{split}
  {\bf{r}}&= \left[ {\begin{array}{*{20}c}
   {\sum\limits_{k = 1,k \ne i}^N {^ \oplus  f_{1 ,k } \tilde x_{k} }  \oplus f_{1 ,i } \tilde  x_{i} }  \\
   {\sum\limits_{k = 1,k \ne i}^N {^ \oplus  f_{2 ,k } \tilde x_{k} }  \oplus f_{2 ,i } \tilde x_{i} }  \\
    \vdots   \\
   {\sum\limits_{k = 1,k \ne i}^N {^ \oplus  f_{N - 1 ,k } \tilde x_{k} }  \oplus f_{N - 1 ,i } \tilde x_{i} }  \\
\end{array}} \right]_{(N-1) \times 1} \\
  &=\left\{ {\left( {{{\bf{F}}_i}{{{\bf{\tilde x}}}_i}} \right)\bmod \left( 2 \right)} \right\} \oplus \left\{ {\left( {{{\bf{f}}_i}{{\tilde x}_i}} \right)\bmod \left( 2 \right)} \right\}
  ,
\end{split}
\end{equation}where \({\bf{F}}_i  = \left[ {\begin{array}{*{20}c}
{{\bf{f}}_1 } & {\cdots} & {{\bf{f}}_{i - 1} } & {{\bf{f}}_{i + 1} } & {\cdots} & {{\bf{f}}_N }
\end{array}} \right]\) is the network sub-encoding matrix of $U_i$,  \( {\bf{ \tilde x}}_i  = \left[
{\begin{array}{*{20}c} {\tilde x_{1} }  & \cdots  & {
\tilde x_{i - 1} } & { \tilde x_{i + 1} } &  \cdots & { \tilde x_{N} }  \\
\end{array}} \right]^T \).

We denote \( {\bf{ \hat x}}_i  = \left[ {\begin{array}{*{20}c} {\hat x_{i,1} }  & \cdots  & {
\hat x_{i,i - 1} } & {\ \hat x_{i,i + 1} } &  \cdots & { \hat x_{i,N} }  \\
\end{array}} \right]^T \) as the bits obtained by network decoding at $U_i$. In the BMNC decoding, based on (\ref{decode
1_1_2}), ${\bf{ \hat x}}_i$ can be obtained through
\begin{equation}
\label{decode 1_2}
\begin{split}
{{\bf{\tilde r}}}_i= \left\{ {\left( {{{\bf{F}}_i}{{{\bf{\hat x}}}_i}} \right)\bmod \left( 2 \right)}
\right\} \oplus \left\{ {\left( {{{\bf{f}}_i}{x_i}} \right)\bmod \left( 2 \right)} \right\}.
\end{split}
\end{equation}

Adding $ \left\{ {\left( {{{\bf{f}}_i}{x_i}} \right)\bmod \left( 2 \right)} \right\} $ on both sides
of the above equation, (\ref{decode 1_2}) can be rewritten as
\begin{equation}
\label{decode 1_3}
\begin{split}
\left( {{{\bf{F}}_i}{{{\bf{\hat x}}}_i}} \right)\bmod \left( 2 \right) = {{{\bf{\tilde r}}}_i} \oplus
{\rm{ }}\left\{ {\left( {{{\bf{f}}_i}{x_i}} \right)\bmod \left( 2 \right)} \right\}.
\end{split}
\end{equation}

If the matrix \({\bf{F}}_i\) is not full rank, the column vector $ \left\{ {\left( {{{\bf{f}}_i}{x_i}}
\right)\bmod \left( 2 \right)} \right\}$ does not have $N-1$ independent elements so that $U_i$ can
not obtain all the information bits of other $N-1$ users. Thus \({\bf{F}}_i\) should be full rank,
then the inverse matrix \({\bf{F}}_i^{ - 1}\) exists. Multiplying \({\bf{F}}_i^{ - 1}\) on the two
sides of (\ref{decode 1_3}), we have
\begin{equation}
\label{decode 1_4}
\begin{split}
{{{\bf{\hat x}}}_i} = \left\{ {{\bf{F}}_i^{ - 1}\left\{ {{{{\bf{\tilde r}}}_i} \oplus \left\{ {\left(
{{{\bf{f}}_i}{x_i}} \right)\bmod \left( 2 \right)} \right\}} \right\}} \right\}\bmod \left( 2 \right).
\end{split}
\end{equation}


It can be seen that when \({\bf{F}}_i\) is full rank, $U_i$ can decode the information through
(\ref{decode 1_4}). If \({\bf{F}}_i\) is not full rank, $U_i$ can not obtain all the information bits
of other users. Thus, \({\bf{F}}_i\) should be full rank for $i \in \left[{1,N}\right]$ to ensure that
all the users can acquire the information of the other users.

Then we propose the following design criterion on the BMNC matrix to achieve the information exchange.
\begin{theorem}
For $U_i$, if \({\bf{F}}_i\) is full rank, then
\begin{equation}
\label{Lemma1}
\begin{split}
{\bf{f}}_i  = \sum\limits_{k = 1,k \ne i}^N { ^ \oplus {\bf{f}}_k }
\end{split}
\end{equation} is the necessary and sufficient
condition that \({\bf{F}}_j\) is full rank for $j \in \left[{1,N}\right]$.
\end{theorem}
\begin{proof}
See Appendix A.
\end{proof}

Using Theorem 1, the network coding matrix $\bf{F}$ can be easily designed through one full rank
matrix in GF(2). Moreover, in the following, Theorem 1 is used for performance analysis.

\section{Performance Analysis}
Above, network encoding and decoding protocols are investigated. From (\ref{decode 1_4}), it is
evident that different $\bf{F}$ results in different system performance. Thus, Theorem 1 is not
sufficient for further performance analysis and improvement. In this section, the BMNC matrix and the
error performance of the system will be analyzed.
\subsection{Network coding matrix analysis}
In what follows, we shall study how the network coding matrix affects error rates at $U_i$. First, we
have the following result:
\begin{theorem} Since $U_i$ needs to obtain other $N-1$ users'
information bits, the error vectors that user $i$ receives are
\begin{equation}
\label{user i error}
\begin{split}
{{{\bf{\hat x}}}_{e,i}} &=  {{{\bf{x}}}_{i}}\oplus {{{\bf{\hat x}}}_{i}} \\
 &=\left( {{{\bf{x}}_i} \oplus {{{\bf{\tilde x}}}_i}} \right) \oplus \left( {\left( {{x_i} \oplus {{\tilde x}_i}} \right){{\bf{1}}_{N - 1 \times 1}}} \right) \oplus \left\{ {\left( {{\bf{F}}_i^{ - 1}\left( {{\bf{r}} \oplus {{{\bf{\tilde r}}}_i}} \right)} \right)\bmod \left( 2 \right)} \right\},
\end{split}
\end{equation}where \( {\bf{x}}_i  = \left[
{\begin{array}{*{20}c}
   {x_1 }  &  \cdots  & {x_{i - 1} } & {x_{i + 1} } &  \cdots  & {x_N }  \\
\end{array}} \right]^T
\) refers to the information that $U_i$ wishes to obtain, \( {{\bf{1}}_{N - 1 \times 1} } \) is a
column vector with $N-1$ elements which are all $1$.
\end{theorem}
\begin{proof}
See Appendix B.
\end{proof}

\subsection{Exact system performance}
In this subsection, closed-form expressions of the SEP and throughput of the system will be derived.

\noindent \textbf{The SEP of the system:}

First, we shall investigate the addition in GF(2) and give the following lemma which will be used for
our results later.
\begin{Lemma}
For addition of $Q$ numbers in GF(2), we have
\begin{equation}
\label{GF(2) add lemma}
\begin{split}
\sum\limits_{q = 1}^Q {^ \oplus  a_q }  = \sum\limits_{q = 1}^Q {{{\left( { - 2} \right)}^{q -
1}}\sum\limits_{1 \le {p_1} < {p_2} <  \cdots  < {p_q} \le Q}^{} {\prod\limits_{j = 1}^q {{a_{{p_j}}}}
} },
\end{split}
\end{equation}where $a_q  \in \left\{ {0,1} \right\}$.
\end{Lemma}
\begin{proof}
See Appendix C.
\end{proof}

Since (\ref{user i error}) is not convenient for SEP analysis, (\ref{user i error}) should be
transformed and the third item of (\ref{user i error}) can be rewritten as
\begin{equation}
\label{user i error_add GF(2) 1}
\begin{split}
 \left\{ {{\bf{F}}_i^{ - 1}\left( {{\bf{r}} \oplus {{{\bf{\tilde r}}}_i}} \right)} \right\}\bmod \left( 2 \right) & = \left\{ {\left[ {\begin{array}{*{20}{c}}
{{{\bf{G}}_{i,1}}}&{{{\bf{G}}_{i,2}}}& \cdots &{{{\bf{G}}_{i,N - 1}}}
\end{array}} \right]\left( {{\bf{r}} \oplus {{{\bf{\tilde r}}}_i}} \right)} \right\}\bmod \left( 2 \right)  \\
  &= \sum\limits_{n = 1}^{N - 1} {^ \oplus  \left( {{\bf{G}}_{i,n}  \circ \left( {{\bf{r}} \oplus {{{\bf{\tilde r}}}_i}} \right) }
  \right)},
\end{split}
\end{equation}where vector \( {{\bf{G}}_{i,k} } \) is the $k$th column vector of
\( {{\bf{F}}_i ^{ - 1} } \).

Substituting (\ref{user i error_add GF(2) 1}) into (\ref{user i error}), we have
\begin{equation}
\label{user i error_add GF(2) 2}
\begin{split}
 {\bf{\hat x}}_{e,i}  &= \left( {{{\bf{x}}_i} \oplus {{{\bf{\tilde x}}}_i}} \right) \oplus \left( {\left( {{x_i} \oplus {{\tilde x}_i}} \right){\bf{1}}_{N - 1 \times 1}  } \right) \oplus \left( {\sum\limits_{n = 1}^{N - 1} {^ \oplus  \left( {{\bf{G}}_{i,n}  \circ \left( {{\bf{r}} \oplus {{{\bf{\tilde r}}}_i}} \right)} \right)} } \right) \\
  &= \sum\limits_{n = 1}^{N + 1} {^ \oplus  {\bf{a}}_{i,n} }  ,
\end{split}
\end{equation}where ${\bf{a}}_{i,1}  = \left( {{{\bf{x}}_i} \oplus {{{\bf{\tilde x}}}_i}} \right)$, ${\bf{a}}_{i,2} ={\left( {{x_i} \oplus {{\tilde x}_i}} \right){\bf{1}}_{N - 1 \times 1}  }
$, ${\bf{a}}_{i,n} = {{\bf{G}}_{i,n-2}  \circ \left( {{\bf{r}} \oplus {{{\bf{\tilde r}}}_i}} \right) }
$, for $3 \le n \le N+1$.

Using Lemma 3, (\ref{user i error_add GF(2) 2}) can be expressed as
\begin{equation}
\label{user i error_add GF(2) 3}
\begin{split}
{\bf{\hat x}}_{e,i}  = \sum\limits_{n = 1}^{N + 1} {{{\left( { - 2} \right)}^{n - 1}}\sum\limits_{1
\le {p_1} < {p_2} <  \cdots  < {p_n} \le N + 1}^{} {\prod\limits_{j = 1}^n {^ \circ
{{\bf{a}}_{i,{p_j}}}} } }.
\end{split}
\end{equation}

Based on (\ref{user i error_add GF(2) 3}), using Bayesian formula,
the error probability of user $i$ can be calculated as
\begin{equation}
\label{user i error_add GF(2) 4}
\begin{split}
 P_{e,i} \left( {\bf{F}} \right) &= E\left[ {\left| {{\bf{\hat x}}_{e,i} } \right|} \right] \\
  &= \sum\limits_{k = 1}^{N - 1} {E\left[ {\left( {{\bf{\hat x}}_{e,i} } \right)_k } \right]}  \\
  &= \sum\limits_{k = 1}^{N - 1} {\sum\limits_{n = 1}^{N + 1} {{{\left( { - 2} \right)}^{n - 1}}\sum\limits_{1 \le {p_1} < {p_2} <  \cdots  < {p_n} \le N + 1}^{} {\prod\limits_{j = 1}^n {E\left[ {{{\left( {{{\bf{a}}_{i,{p_j}}}} \right)}_k}} \right]} } } }
  ,
\end{split}
\end{equation}where
\begin{equation}
\label{user i error_add GF(2) 5}
\begin{split}
E\left[ {\left( {{\bf{a}}_{i,1} } \right)_k } \right] &= \left\{
{\begin{array}{ll}
   {{E\left[ { {{x_k} \oplus {{\tilde x}_k}} } \right] } ,} & {{k \le i - 1,}} \\
   {{E\left[ {{{x_{k+1}} \oplus {{\tilde x}_{k+1}}} } \right] ,}} & {{ i \le k \le N-1,} } \\
\end{array}} \right. \\
E\left[ {\left( {{\bf{a}}_{i,2}  } \right)_k } \right] &=E\left[{  {{x_i} \oplus {{\tilde x}_i}} }
\right] ,\\
E\left[ {\left( {{\bf{a}}_{i,l} } \right)_k } \right] &=E\left[ { {{r_k} \oplus {{\tilde r}_{i,k}}} }
\right]\left( {{\bf{F}}_i^{ - 1} } \right)_{k,l - 2}
 , \ {\rm{for
}} \ 3 \le l \le N+1 .
\end{split}
\end{equation}

For BPSK, the relationship between the SEP and the received SNR over Rayleigh fading channels is
\cite{rvq:digital_com_over_fading_channels}
\begin{equation}
\label{SER BPSK}
\begin{split}
P\left( {e_\Delta  } \right) = \frac{1}{2} - \frac{1}{2}\sqrt
{\frac{{\gamma _\Delta  }}{{1 + \gamma _\Delta  }}} ,
\end{split}
\end{equation}where $e_\Delta $ denotes the error and $ \gamma _\Delta$ denotes the average received
SNR. Thus, (\ref{user i error_add GF(2) 5}) can be rewritten as
\begin{equation}
\label{user i error_add GF(2) 6}
\begin{split}
E\left[ {\left( {{\bf{a}}_{i,1} } \right)_k } \right] &= \left\{
{\begin{array}{ll}
   {\frac{1}{2} - \frac{1}{2}\sqrt {\frac{{\gamma _{x_k } }}{{1 + \gamma
_{x_k } }}} ,} & {{k \le i - 1,}} \\
   {{\frac{1}{2} - \frac{1}{2}\sqrt {\frac{{\gamma _{x_{k+1} } }}{{1 + \gamma
_{x_{k+1} } }}},}} & {{ i \le k \le N-1,} }
\end{array}} \right. \\
E\left[ {\left( {{\bf{a}}_{i,2} } \right)_k } \right] &=
\frac{1}{2} -
\frac{1}{2}\sqrt {\frac{{\gamma _{x_i } }}{{1 + \gamma _{x_i } }}},\\
E\left[ {\left( {{\bf{a}}_{i,l} } \right)_k } \right] &= \left( {\frac{1}{2} - \frac{1}{2}\sqrt
{\frac{{\gamma _{i,r_k } }}{{1 + \gamma _{i,r_k } }}} } \right) \left( {{\bf{F}}_i^{ - 1} }
\right)_{k,l - 2}
 , \ {\rm{for
}} \ 3 \le l \le N+1,
\end{split}
\end{equation}where $\gamma _{x_k }$ is the average received SNR of $U_k$ at
the relay and $\gamma _{i,r_k }$ is the average received SNR at $U_i$ in time slot $k$ in the second
phase.

Thus the error probability of the system can be expressed as
\begin{equation}
\label{exact N user error probability 1}
\begin{split}
 P_e \left( {\bf{F}} \right) &= \frac{1}{N}\sum\limits_{i = 1}^N {P_{e,i} \left( {\bf{F}} \right)} ,
\end{split}
\end{equation}where ${P_{e,i} \left( {\bf{F}} \right)}$ is given in (\ref{user i error_add GF(2) 4}).

\noindent \textbf{The throughput of the system:}

Here we define the throughput as the symbols received correctly at all the users per time slot. Then
the throughput of the system with NC is given by
\begin{equation}
\label{throughput}
\begin{split}
 \mathbb{T}_{NC} = \frac{{N\left( {N - 1} \right)\left( {1 - {P_e}\left( {\bf{F}} \right)} \right)}}{{2N -
 1}},
 \end{split}
\end{equation}where $N$ denotes the number of the users and $P_e \left( {\bf{F}} \right) $ is given in
(\ref{exact N user error probability 1}).

As the extension of two-way relay network \cite{MUNC:two-way_liyouhui}, for the system without NC, we
first derive the SEP of the system. For one user, the error probability of other users receiving $x_i$
can be expressed as
\begin{align}\label{A 10 1}
{P_{e,{x_i}}} = (N-1){P_e}\left( {{x_i}} \right) + \left( {1 - {P_e}\left( {{x_i}} \right)}
\right)\sum\limits_{j = 1}^{N - 1} {{P_e}\left( {{r_{j,i}}} \right)}
\end{align}where ${P_e}\left( {{x_i}} \right)$ denotes the error probability of the relay detecting
$x_i$ and ${{P_e}\left( {{r_{j,i}}} \right)}$ denotes the error probability of $U_j$ detecting $r_i$.
In the system without NC, it is evident that $r_i$ is $x_i$.

Then the SEP of the system without NC can be expressed as
\begin{align}\label{Sep no nc}
{P_e}' &= \sum\limits_{i = 1}^N {\left\{ {(N-1){P_e}\left( {{x_i}} \right) + \left( {1 - {P_e}\left( {{x_i}} \right)} \right)\sum\limits_{j = 1}^{N - 1} {{P_e}\left( {{r_{j,i}}} \right)} } \right\}} \notag\\
 &= \sum\limits_{i = 1}^N {\left\{ {\frac{1}{2}(N-1)\left( {1 - \sqrt {\frac{{{\gamma _{{x_i}}}}}{{1 + {\gamma _{{x_i}}}}}} } \right) + \frac{1}{2}\left( {1 - \frac{1}{2}\left( {1 - \sqrt {\frac{{{\gamma _{{x_i}}}}}{{1 + {\gamma _{{x_i}}}}}} } \right)} \right)\sum\limits_{j = 1}^{N - 1} {\left( {1 - \sqrt {\frac{{{\gamma _{j,{r_i}}}}}{{1 + {\gamma _{j,{r_i}}}}}} } \right)} } \right\}}
\end{align}where Eq. (19) is used in the last step.

The throughput of the system without NC can be easily obtained as
\begin{align}\label{throughput_no nc}
\mathbb{T}_{No \ NC} =\frac{{\left( {N - 1} \right)\left( {1 - {P_e}^\prime } \right)}}{{2}},
\end{align}where ${P_e}^\prime$ is given in (\ref{Sep no nc}).

In the high SNR region, the throughput of the scheme with NC and without NC can be expressed as
\begin{align}\label{B 1}
\mathbb{T}{_{NC}^\infty} &= \mathop {\lim }\limits_{SNR \to \infty } \mathbb{T}{_{NC}}\notag\\
 &= \mathop {\lim }\limits_{Pe\left( {\bf{F}} \right) \to 0} \mathbb{T}{_{NC}}\notag\\
 &= \frac{{N\left( {N - 1} \right)}}{{2N - 1}}.
\end{align}
\begin{align}\label{B 2}
\mathbb{T}{_{No \ \! NC}^\infty} &= \mathop {\lim }\limits_{SNR \to \infty } \mathbb{T}{_{No\ \! NC}}\notag\\
 &= \frac{{ {N - 1} }}{{2}}.
\end{align}where $\mathbb{T}{_{NC}^\infty}$ and $\mathbb{T}{_{No \ \! NC}^\infty}$ denote
the throughput of the system with and without NC in the high SNR region, respectively.

Based on (\ref{B 1}) (\ref{B 2}), the absolute value of the throughput improvement equals to
\begin{align}\label{B 3}
\mathbb{T}{_\Delta ^\infty} &= \mathbb{T}{_{NC}^\infty} - \mathbb{T}{_{No \ \! NC}^\infty}\notag\\
 &= \frac{{N\left( {N - 1} \right)}}{{2N - 1}} - \frac{{ {N - 1}}}{{2}}\notag\\
 &= \frac{1}{4}\left( {1 - \frac{1}{{2N - 1}}} \right).
\end{align} It can be seen that $\mathbb{T}{_\Delta ^\infty}$ increases with $N$,
which indicates that increasing the number of the users brings an improved performance of the absolute
value of the throughput.



\subsection{System performance bound}
Above, the exact closed-form expressions of SEP and throughput of the system have been derived.
However, it can be seen that the expressions are very complex and provide few insights. In this
subsection, the tight upper bound of SEP of the system will be derived to show useful insights.
\begin{Lemma}
For addition and multiplication in GF(2), we have
\begin{equation}
\label{GF(2) add}
\begin{split}
a \oplus b \le a+b ,
\end{split}
\end{equation}
\begin{equation}
\label{GF(2) multiply}
\begin{split}
\left( {{\bf{AB}}} \right)\bmod \left( 2 \right) \le {\bf{A}}  {\bf{B}} ,
\end{split}
\end{equation}where ${\bf{A}}$ is a $\left({L  \times M}\right)$ matrix and ${\bf{B}}$ is a column vector with $M$
elements. Symbol $a$ and $b$, all the elements in ${\bf{A}}$ and ${\bf{B}}$ are in GF(2).
\end{Lemma}
\begin{proof}
See Appendix D.
\end{proof}



Using Lemma 4, (\ref{user i error}) can be upper bounded as
\begin{equation}
\label{user i error upper bound 1}
\begin{split}
{{{\bf{\hat x}}}_{e,i}} \le \left( {{{\bf{x}}_i} \oplus {{{\bf{\tilde x}}}_i}} \right) + \left(
{{{\bf{1}}_{N - 1 \times 1}}\left( {{x_i} \oplus {{\tilde x}_i}} \right)} \right) + \left(
{{\bf{F}}_i^{ - 1}\left( {{\bf{r}} \oplus {{{\bf{\tilde r}}}_i}} \right)} \right)\triangleq  {\bf{\hat
x}}_{e,i}^U,
\end{split}
\end{equation}where in the high SNR region, the probability that two errors occur simultaneously is much lower than the
probability that only one error occurs. Thus, the condition that more than two errors occur
simultaneously can be ignored. We note that this simplified bound is still very tight as we can see
from following simulations.

Based on (\ref{user i error upper bound 1}), the error probability of user $i$ can be calculated as
\begin{equation}
\label{user i error probability}
\begin{split}
 P_{e,i} \left( {\bf{F}} \right) &\le E\left[ {\left| {{\bf{\hat x}}_{e,i}^U } \right|} \right] \\
  &= \sum\limits_{k = 1,k \ne i}^N {E\left[ {{{x_k} \oplus {{\tilde x}_k}}  } \right]}  + \left( {N - 1} \right)E\left[ {  {{x_i} \oplus {{\tilde x}_i}}  } \right] + \sum\limits_{k = 1}^{N - 1} {E\left[ { {{r_k} \oplus {{\tilde r}_{i,k}}}  } \right]} \left| {{\bf{G}}_{i,k} } \right|
  \triangleq  P_{e,i}^U \left( {\bf{F}} \right),
\end{split}
\end{equation}

Using (\ref{user i error probability}), the SEP of the system can be upper bounded as
\begin{equation}
\label{SER system}
\begin{split}
P_e \left( {\bf{F}} \right) & \le \frac{1}{N}\sum\limits_{i = 1}^N {P_{e,i}^U \left( {\bf{F}} \right)}
  \triangleq  P_{e}^U \left( {\bf{F}} \right).
\end{split}
\end{equation}

Substituting (\ref{SER BPSK}) into (\ref{SER system}), the final upper bound of the SEP of the system
is
\begin{equation}
\label{SER system_1}
\begin{split}
 P_e^U \left( {\bf{F}} \right) &= \frac{1}{N}\sum\limits_{i = 1}^N {\left\{ {\sum\limits_{k = 1,k \ne i}^N {\left( {\frac{1}{2} - \frac{1}{2}\sqrt {\frac{{\gamma _{x_k } }}{{1 + \gamma _{x_k } }}} } \right)} } \right.}  \\
 & \quad \left. { + \left( {N - 1} \right)\left( {\frac{1}{2} - \frac{1}{2}\sqrt {\frac{{\gamma _{x_i } }}{{1 + \gamma _{x_i } }}} } \right) + \sum\limits_{k = 1}^{N - 1} {\left( {\frac{1}{2} - \frac{1}{2}\sqrt {\frac{{\gamma _{i,r_k } }}{{1 + \gamma _{i,r_k } }}} } \right)} \left| {{\bf{G}}_{i,k} } \right|} \right\}
 .
\end{split}
\end{equation}From the above expression (\ref{SER system_1}), it can be seen that for $U_i$, the received error probability
is made up of three parts: the first item ${\sum\limits_{k = 1,k \ne i}^N {P\left( {e_{x_k } }
\right)}}$ results from the transmissions of $N-1$ users except $U_i$ in the first phase, the second
item $\left( {N - 1} \right)P\left( {e_{x_i } } \right)$ results from the transmission of $U_i$ itself
in the first phase, the third item $\sum\limits_{k = 1}^{N - 1} {P\left( {e_{i,r_k } } \right)} \left|
{{\bf{G}}_{i,k} } \right|$ results from the transmissions of relay in the second phase. It is evident
that the impacts of the three items on the SEP of the system are on the same order of magnitude, since
their coefficients are all about $N-1$. It can also be seen that the third item has the largest impact
on the SEP of the system, since $\left| {{\bf{G}}_{i,k} } \right| \ge 1$. As only one relay is
employed to assist the users and the users do not cooperative with each other in the system, the
diversity order of the proposed scheme is $1$.

\subsection{Designed BMNC encoding matrix}
Above, the connection between the system error performance and the network coding matrix is provided.
Moreover, several design criteria of the network coding matrix, which ensure the successful
information exchange, are also investigated. It can be seen from (\ref{SER system_1}) that the network
coding matrix has a significant impact on the system performance. Thus the network coding matrix
should be designed carefully to further improve the system error performance.

In practical systems, the distance between the relay and the user varies for different users. Then the
average received SNRs at the relay for different users may not be the same. Moreover, for high order
modulations, the power that the relay uses also varies for different symbols.

It is assumed that the statistical channel state information is known at the relay. Without loss of
generality, we assume that the statistic channel conditions between the users and the relay have an
ascending order from $U_1$ to $U_N$, which means that the statistic distance between the $U_i$ and the
relay is larger than that between the $U_j$ and the relay when $i<j$. Moreover, we assume that the
power that the relay uses to broadcast the detected information has a descending order from time slot
$1$ to time slot $N-1$, which means that the power that the relay uses in time slot $i$ is higher than
that in time slot $j$ when $i<j$. In the following, we will show that the assumed order of the channel
gains and that of the power allocation formulate the designed network coding matrix in a more detail.

To simplify the network coding design process and save the memory, we propose a puncturing operation
based network coding matrix design scheme. In this scheme, the designed encoding matrix of $N$ users
is a matrix in which the upper left corner is the designed encoding matrix of $N-1$ users. Thus, the
relay only needs to memorize the designed network coding matrix of the maximum number users.

We design an encoding matrix for $N$ users as
\begin{equation}
\label{N user matrix}
\begin{split}
{\bf{F}}_{|N \ users}=\left[ {\begin{array}{*{20}c}
   1 & 1 & 0 &  \cdots  & 0  \\
   1 & 0 & 1 &  \cdots  & 0  \\
    \vdots  &  \vdots  &  \vdots  &  \ddots  &  \vdots   \\
   1 & 0 & 0 &  \cdots  & 1  \\
\end{array}} \right]_{(N-1) \times N},
\end{split}
\end{equation}the optimality of which will be discussed in Section V.

\section{The Optimality of the Matrix Given in (\ref{N user matrix})}
In this section, we will discuss the optimality of the matrix given in (\ref{N user matrix}). First,
we consider the situation that there are three users.
\begin{Lemma}
For $N = 3$, the designed network coding matrix $\bf{F}$, which minimizes the bound of SEP, can be
designed as follows
\begin{equation}
\label{3 user matrix}
\begin{split}
{\bf{F}}_{|3 \ users}=\left[ {\begin{array}{*{20}c}
   1 & 1 & 0  \\
   1 & 0 & 1  \\
\end{array}} \right].
\end{split}
\end{equation}
\end{Lemma}

\begin{proof}
See Appendix E.
\end{proof}

In the following, the designed network coding matrices for the systems with more than three users will
be discussed. We suppose the designed encoding matrix of $N-1$ users is a ${\left({N-2}\right) \times
\left({N-1}\right)}$ matrix, which is evaluated by
\begin{equation}
\label{N-1 user matrix}
\begin{split}
{\bf{F}}_{|N-1 \ users}=\left[ {\begin{array}{*{20}c}
   1 & 1 & 0 &  \cdots  & 0  \\
   1 & 0 & 1 &  \cdots  & 0  \\
    \vdots  &  \vdots  &  \vdots  &  \ddots  &  \vdots   \\
   1 & 0 & 0 &  \cdots  & 1  \\
\end{array}} \right]_{\left({N-2}\right) \times
\left({N-1}\right)}.
\end{split}
\end{equation}

Based on this assumption and Lemma 5, if the designed encoding matrix of $N$ users is the matrix
described in (\ref{N user matrix}), using mathematical induction, the matrix given in (\ref{N user
matrix}) is the matrix we want. Using (\ref{N-1 user matrix}) and the proposed network coding matrix
design scheme, the designed encoding matrix of $N$ users can be written as
\begin{equation}
\label{N user matrix_x}
\begin{split}
{\bf{F}}_{|N \ users}=\left[ {\begin{array}{*{20}c}
   1 & 1 & 0 &  \cdots  & 0  & b_1\\
   1 & 0 & 1 &  \cdots  & 0  & b_2\\
    \vdots  &  \vdots  &  \vdots  &  \ddots  &  \vdots  &  \vdots  \\
   1 & 0 & 0 &  \cdots  & 1  & b_{N-2}\\
   u_1 & u_2 & u_3 &  \cdots  & u_{N-1} & u_{N} \\
\end{array}} \right]_{(N-1) \times N},
\end{split}
\end{equation}where $b_1$, $b_2$, $\ldots$, $b_{N-2}$, \( {u_1 } \), \( {u_2 } \), $\ldots$, \( {u_{N} }
\) are unknown elements, which are in GF(2).

Using Theorem 1, each element of the last column vector, $b_1$, $b_2$, $\ldots$, $b_{N-2}$, should be
the sum of the other elements in its row vector in GF(2). For the first $N-2$ row vectors, since there
are just two ``$1$" elements expect the last column in one row, we have $b_1=b_2=\ldots=b_{N-2}=0$.

Since \( {{\bf{F}}_{i} } \) should be full rank for any $i$, each column vector of \( {{\bf{F}}_{i} }
\) should not be a zero column vector. Thus, in the last column vector, \( {u_{N} } \) should be $1$,
since the other elements in this column vector are all $0$. The designed encoding matrix of $N$ users
can be rewritten as
\begin{equation}
\label{N user matrix_1}
\begin{split}
{\bf{F}}_{|N \ users}=\left[ {\begin{array}{*{20}c}
   1 & 1 & 0 &  \cdots  & 0  & 0\\
   1 & 0 & 1 &  \cdots  & 0  & 0\\
    \vdots  &  \vdots  &  \vdots  &  \ddots  &  \vdots  &  \vdots  \\
   1 & 0 & 0 &  \cdots  & 1  & 0\\
   u_1 & u_2 & u_3 &  \cdots  & u_{N-1} & 1 \\
\end{array}} \right]_{(N-1) \times N},
\end{split}
\end{equation}where using Theorem 1, we have
\begin{equation}
\label{uk relationship}
\begin{split}
{\sum\limits_{ k = 1}^{N - 1} {^\oplus u_k } =1}.
\end{split}
\end{equation}

Thus, only \( {u_1 } \), \( {u_2 } \), $\ldots$, \( {u_{N-1} } \), which are the elements of the last
row vector, are left to be designed to improve the system error performance. From (\ref{SER
system_1}), it can be seen that the network decoding matrices have considerable impact on the SEP of
the system. In the following, we need to acquire the network decoding matrix ${\bf{F}}_i^{ - 1}$ for
$i \in \left[ {1,N } \right]$. Elementary row operations are used to obtain the network decoding
matrix.

For the first user, using elementary row operations, we have
\begin{equation}
\label{N user H_1 1}
\begin{split}
 [\begin{array}{*{20}c}
   {{\bf{F}}_{1| N \ users} } & {\bf{I}}  \\
\end{array}] &= \left[ {\left. {\begin{array}{*{20}c}
   1 & 0 &  \cdots  & 0 & 0  \\
   0 & 1 &  \cdots  & 0 & 0  \\
    \vdots  &  \vdots  &  \ddots  &  \vdots  &  \vdots   \\
   0 & 0 &  \cdots  & 1 & 0  \\
   {u_2 } & {u_3 } &  \cdots  & {u_{N - 1} } & 1  \\
\end{array}} \right|\begin{array}{*{20}c}
   1 & 0 &  \cdots  & 0 & 0  \\
   0 & 1 &  \cdots  & 0 & 0  \\
    \vdots  &  \vdots  &  \ddots  &  \vdots  &  \vdots   \\
   0 & 0 &  \cdots  & 1 & 0  \\
   0 & 0 &  \cdots  & 0 & 1  \\
\end{array}} \right]_{(N - 1) \times (2N - 2)}  \\
  &\to \left[ {\left. {\begin{array}{*{20}c}
   1 & 0 &  \cdots  & 0 & 0  \\
   0 & 1 &  \cdots  & 0 & 0  \\
    \vdots  &  \vdots  &  \ddots  &  \vdots  &  \vdots   \\
   0 & 0 &  \cdots  & 1 & 0  \\
   0 & 0 &  \cdots  & 0 & 1  \\
\end{array}} \right|\begin{array}{*{20}c}
   1 & 0 &  \cdots  & 0 & 0  \\
   0 & 1 &  \cdots  & 0 & 0  \\
    \vdots  &  \vdots  &  \ddots  &  \vdots  &  \vdots   \\
   0 & 0 &  \cdots  & 1 & 0  \\
   {u_2 } & {u_3 } &  \cdots  & {u_{N - 1} } & 1  \\
\end{array}} \right]_{(N - 1) \times (2N - 2)}  ,
\end{split}
\end{equation}where in the last step, the first $N-2$ row vectors are multiplied by
different coefficient, for example the $j$th row vector is multiplied by $u_{j+1}$ for $j \le N-2$,
and then are all added to the last row vector. From (\ref{N user H_1 1}), the inverse matrix of \(
{{\bf{F}}_{1} } \) is
\begin{equation}
\label{N user H_1_-1}
\begin{split}
{\bf{F}}_{1| N \ users}^{ - 1}  =\left[ {\begin{array}{*{20}c}
    1 & 0 &  \cdots  & 0  & 0\\
    0 & 1 &  \cdots  & 0  & 0\\
      \vdots  &  \vdots  &  \ddots  &  \vdots  &  \vdots  \\
    0 & 0 &  \cdots  & 1  & 0\\
    u_2 & u_3 &  \cdots  & u_{N-1} & 1 \\
\end{array}} \right]_{(N-1) \times (N-1)}= {\bf{F}}_1 .
\end{split}
\end{equation}

For the second user, using elementary row operations, we have
\begin{equation}
\label{N user H_2_2}
\begin{split}
 \left[ {\begin{array}{*{20}c}
   {{\bf{F}}_{2|N \ users} } & {\bf{I}}  \\
\end{array}} \right] = \left[ {\left. {\begin{array}{*{20}c}
   1 & 0 &  \cdots  & 0 & 0  \\
   1 & 1 &  \cdots  & 0 & 0  \\
    \vdots  &  \vdots  &  \ddots  &  \vdots  &  \vdots   \\
   1 & 0 &  \cdots  & 1 & 0  \\
   {u_1 } & {u_3 } &  \cdots  & {u_{N - 1} } & 1  \\
\end{array}} \right|\begin{array}{*{20}c}
   1 & 0 &  \cdots  & 0 & 0  \\
   0 & 1 &  \cdots  & 0 & 0  \\
    \vdots  &  \vdots  &  \ddots  &  \vdots  &  \vdots   \\
   0 & 0 &  \cdots  & 1 & 0  \\
   0 & 0 &  \cdots  & 0 & 1  \\
\end{array}} \right]_{(N - 1) \times (2N - 2)}  \\
  \to \left[ {\left. {\begin{array}{*{20}c}
   1 & 0 &  \cdots  & 0 & 0  \\
   0 & 1 &  \cdots  & 0 & 0  \\
    \vdots  &  \vdots  &  \ddots  &  \vdots  &  \vdots   \\
   0 & 0 &  \cdots  & 1 & 0  \\
   0 & {u_3 } &  \cdots  & {u_{N - 1} } & 1  \\
\end{array}} \right|\begin{array}{*{20}c}
   1 & 0 &  \cdots  & 0 & 0  \\
   1 & 1 &  \cdots  & 0 & 0  \\
    \vdots  &  \vdots  &  \ddots  &  \vdots  &  \vdots   \\
   1 & 0 &  \cdots  & 1 & 0  \\
   {u_1 } & 0 &  \cdots  & 0 & 1  \\
\end{array}} \right]_{(N - 1) \times (2N - 2)}  ,
\end{split}
\end{equation}where in the last step, the first row vector is added directly to other $N-3$ row
vectors except the last row vector, which is added by the first row vector multiplied by $u_1$. Then
for $2 \le j \le N-2$, multiplying the $j$th row vector with $u_{j+1}$ and adding the product to the
last row vector, (\ref{N user H_2_2}) can be rewritten as
\begin{equation}
\label{N user H_2_3}
\begin{split}
 \left[ {\begin{array}{*{20}c}
   {{\bf{F}}_{2|N\ users} } & {\bf{I}}  \\
\end{array}} \right]
 \to \left[ {\left. {\begin{array}{*{20}c}
   1 & 0 &  \cdots  & 0 & 0  \\
   0 & 1 &  \cdots  & 0 & 0  \\
    \vdots  &  \vdots  &  \ddots  &  \vdots  &  \vdots   \\
   0 & 0 &  \cdots  & 1 & 0  \\
   0 & 0 &  \cdots  & 0 & 1  \\
\end{array}} \right|\begin{array}{*{20}c}
   1 & 0 &  \cdots  & 0 & 0  \\
   1 & 1 &  \cdots  & 0 & 0  \\
    \vdots  &  \vdots  &  \ddots  &  \vdots  &  \vdots   \\
   1 & 0 &  \cdots  & 1 & 0  \\
   {\sum {_{k = 1,k \ne 2}^{N - 1} {^ \oplus} u_k } } & {u_3 } &  \cdots  & {u_{N - 1} } & 1  \\
\end{array}} \right]_{(N - 1) \times (2N - 2)} .
\end{split}
\end{equation}

From (\ref{N user H_2_3}), the inverse matrix of \( {{\bf{F}}_{2} } \) can be expressed as
\begin{equation}
\label{N user H_2_-1}
\begin{split}
{\bf{F}}_{2|N \ users}^{-1}=\left[ {\begin{array}{*{20}c}
    1 & 0 &  \cdots  & 0  & 0\\
    1 & 1 &  \cdots  & 0  & 0\\
      \vdots  &  \vdots  &  \ddots  &  \vdots  &  \vdots  \\
    1 & 0 &  \cdots  & 1  & 0\\
{\sum {_{k = 1,k \ne 2}^{N - 1} {^ \oplus}  u_k } }
 & u_3 &  \cdots  & u_{N-1} & 1 \\
\end{array}} \right]_{(N-1) \times (N-1)} .
\end{split}
\end{equation}

In the same way as the inverse matrix of \( {{\bf{F}}_{1} } \) and \( {{\bf{F}}_{2} } \), the inverse
matrix of \( {{\bf{F}}_{i} } \) for $i \in \left[ {3,N-1 } \right]$ can be described as
\begin{equation}
\label{N user H_i_-1}
\begin{split}
{\bf{F}}_{i| N \ users}^{ - 1}  = \left[ {\begin{array}{*{20}c}
   0       &     0    & \cdots   &       0      &    1   &  0    &  \cdots     & 0       & 0\\
   1       &     0    & \cdots   &       0      &    1   &  0    &  \cdots     & 0       & 0\\
   0       &     1    & \cdots   &       0      &    1   &  0    &  \cdots     & 0       & 0\\
   \vdots  &  \vdots  & \ddots   &  \vdots      &\vdots  &\vdots &  \ddots     & \vdots  &\vdots\\
   0       &     0    & \cdots   &       1      &    1   &  0    &  \cdots     & 0       & 0\\
   0       &     0    & \cdots   &       0      &    1   &  0    &  \cdots     & 0       & 0\\
   0       &     0    & \cdots   &       0      &    1   &  1    &  \cdots     & 0       & 0\\
   \vdots  &  \vdots  & \ddots   &  \vdots      &\vdots  &\vdots &  \ddots     & \vdots  &\vdots\\
   0       &     0    & \cdots   &     0        & 1      &  0    &   \cdots     & 1        & 0\\
   u_2     & u_3      & \cdots   &  u_{i-1}  &\!{\sum\limits_{ k = 1,k \ne i}^{N - 1} {^\oplus u_k } }\!& u_{i+1}  &   \cdots     & u_{N-1}  & 1 \\
\end{array}} \right]_{(N-1) \times (N-1)},
\end{split}
\end{equation}and the inverse matrix of \( {{\bf{F}}_{N} } \) equals to
\begin{equation}
\label{N user H_N-1}
\begin{split}
{\bf{F}}_{N| N\ users }^{ - 1}  = \left[ {\begin{array}{*{20}c}
   0      & 0       &  \cdots   &   0      & 1  \\
   1      & 0       &  \cdots   &   0      & 1  \\
   0      & 1       &  \cdots   &   0      & 1  \\
   \vdots &\vdots   &  \ddots   &  \vdots  &  \vdots    \\
   0      & 0       &  \cdots   &   1      & 1  \\
   u_2    & u_3     &  \cdots        & u_{N-1} & 1 \\
\end{array}} \right]_{(N-1) \times (N-1)} .
\end{split}
\end{equation}

We design a matrix ${\bf{\mathord{\buildrel{\lower3pt\hbox{$\scriptscriptstyle\frown$}} \over F}
}}_{|N \ users}$, which can be shown in (\ref{N user matrix}). In the following, we will prove that
using ${\bf{\mathord{\buildrel{\lower3pt\hbox{$\scriptscriptstyle\frown$}} \over F} }}_{|N \ users}$
as the network coding matrix, the upper bound of the SEP of the system is minimized.

Based on (\ref{N user H_1_-1}), (\ref{N user H_2_-1}), (\ref{N user
H_i_-1}) and (\ref{N user H_N-1}), using (\ref{SER system}), we have
\begin{equation}
\label{Pe H_- PeH' }
\begin{split}
 &N\left( {P_e \left( {\bf{F}}_{|N \ users} \right) - P_e \left( {{\bf{\mathord{\buildrel{\lower3pt\hbox{$\scriptscriptstyle\frown$}}
\over F} }}} _{|N \ users}\right)} \right)\\
  &= \sum\limits_{k = 1}^{N - 2} {E\left[ { {{r_k} \oplus {{\tilde r}_{1,k}}}  } \right]} u_{k + 1}
  +{E\left[ { {{r_1} \oplus {{\tilde r}_{2,1}}}  } \right]} \left({ {\sum \limits_{k = 1,k \ne 2}^{N - 1} {{^ \oplus}u_k }}} -1\right) +\sum\limits_{k = 2}^{N - 2} {E\left[ { {{r_k} \oplus {{\tilde r}_{2,k}}}  } \right]} u_{k + 1}
  \\
  &\quad + \sum\limits_{i = 3}^{N - 1} {\left\{ {\sum\limits_{k = 1,k \ne i - 1}^{N - 2} {E\left[ { {{r_k} \oplus {{\tilde r}_{i,k}}}  } \right]u_{k + 1} }  + E\left[ { {{r_{i-1}} \oplus {{\tilde r}_{i,i-1}}}  } \right]\left( { {\sum\limits_{k = 1,k \ne i }^{N - 1} {^\oplus u_k } }  - 1} \right)}
  \right\}}\\
  &\quad + \sum\limits_{k = 1}^{N - 2} {E\left[ { {{r_k} \oplus {{\tilde r}_{N,k}}}  } \right]} u_{k + 1}  .
\end{split}
\end{equation}

Using (\ref{uk relationship}), we have ${ {\sum\limits_{k = 1,k \ne
i }^{N - 1} {^\oplus u_k } }  - 1} =-u_i$ and (\ref{Pe H_- PeH' })
can be expressed as
\begin{equation}
\label{Pe H_- PeH' _1}
\begin{split}
&N\left( {{P_e}\left( {{{\bf{F}}_{|N\;users}}} \right) - {P_e}\left(
{{{{\bf{\mathord{\buildrel{\lower3pt\hbox{$\scriptscriptstyle\frown$}}
\over F} }}}_{|N\;users}}} \right)} \right) \\&= \sum\limits_{k = 1}^{N - 2} {E\left[ {{r_k} \oplus {{\tilde r}_{i,k}}} \right]} {u_{k + 1}} - E\left[ {{r_1} \oplus {{\tilde r}_{2,1}}} \right]{u_2} + \sum\limits_{k = 2}^{N - 2} {E\left[ {{r_k} \oplus {{\tilde r}_{2,k}}} \right]} {u_{k + 1}}\\
 & \quad + \sum\limits_{i = 3}^{N - 1} {\left\{ {\sum\limits_{k = 1,k \ne i - 1}^{N - 2} {E\left[ {{r_k} \oplus {{\tilde r}_{i,k}}} \right]{u_{k + 1}}}  - E\left[ {{r_{i - 1}} \oplus {{\tilde r}_{i,i - 1}}} \right]{u_i}} \right\}}  + \sum\limits_{k = 1}^{N - 2} {E\left[ {{r_k} \oplus {{\tilde r}_{N,k}}} \right]} {u_{k + 1}}\\
 &= \sum\limits_{k = 1}^{N - 2} {\left( {E\left[ {{r_k} \oplus {{\tilde r}_{1,k}}} \right] - E\left[ {{r_k} \oplus {{\tilde r}_{k + 1,k}}} \right]} \right){u_{k + 1}}}  + \sum\limits_{k = 2}^{N - 2} {E\left[ {{r_k} \oplus {{\tilde r}_{2,k}}} \right]} {u_{k + 1}}\\
 &\quad + \sum\limits_{i = 3}^{N - 1} {\sum\limits_{k = 1,k \ne i - 1}^{N - 2} {E\left[ {{r_k} \oplus {{\tilde r}_{i,k}}} \right]{u_{k + 1}}} }  + \sum\limits_{k = 1}^{N - 2} {E\left[ {{r_k} \oplus {{\tilde r}_{N,k}}} \right]} {u_{k + 1}}\\
 &\ge 0,
\end{split}
\end{equation}where $ E\left[ { {{r_k} \oplus {{\tilde r}_{1,k}}}  }
\right] > E\left[ { {{r_k} \oplus {{\tilde r}_{k+1,k}}}  } \right]$ for $ k \in \left[ {1,N - 2}
\right]$ is used in the last inequality. From the above expression, since the coefficients of $u_k$'s
are all strictly positive, it can be seen that the requirement of \( N\left( {P_e \left( {\bf{F}} _{|N
\ users}\right) - P_e \left( {{\bf{\mathord{\buildrel{\lower3pt\hbox{$\scriptscriptstyle\frown$}}
\over F} }}}_{|N \ users} \right)} \right)=0\) is that $u_k=0, k=2,3, \ldots, N-1$, which indicates
that only \({\bf{\mathord{\buildrel{\lower3pt\hbox{$\scriptscriptstyle\frown$}} \over F} }}_{|N \
users}\) meets the requirement. Thus the proposed encoding matrix for $N$ users is unique.

Thus \( {{\bf{\mathord{\buildrel{\lower3pt\hbox{$\scriptscriptstyle\frown$}} \over F} }}} \) is the
designed encoding matrix for $N$ users, which minimizes the upper bound of the SEP of the system.
Based on Lemma 5, using mathematical induction, we have the following theorem.

\begin{theorem}
The designed network coding matrix of $N$ users, which minimizes the bound of SEP, is given in (\ref{N
user matrix}).
\end{theorem}

From Theorem 6, it can be seen that the designed network coding matrix is structured and sparse. The
properties of the matrix simplify the encoding and decoding process while improving system
performance.

In practical systems, the relay first needs to know the number of the users. Moreover, the statistical
information of the user to relay channels should be available to the relay. After receiving the
necessary information, the designed NC matrix is constructed based on its closed-form expression given
by Theorem 6. According to the designed NC matrix, the users will send their information to the relay
in turn. The relay will detect, encode, and broadcast the received information. Finally, each user
decodes its received information by exploiting its own information.

\section{Simulation results}
In this section, the performance of the analytical results will be compared with Monte Carlo
simulations.

As discussed in the Section IV-C, in practical systems, the received power at the relay for different
users is different and so does the transmit power of the relay at different time slot. The conditions
of the simulations are set as follows: the transmit power at different users is the same and the
average received SNR of $U_i$ at the relay is $3$dB worse than that of $U_{i+1}$, due to the different
distances between the users and the relay; the power that the relay uses in time slot $i$ is $3$dB
higher than that in time slot $i+1$, since the power gap between two adjacent bits is $3$dB in some
high order modulations; the smallest transmit power at the relay and the transmit power at one user
are the same; the smallest received SNR at the $U_1$ in the second phase of the system with NC is
denoted as $E_s / N_0$ in the following figures; BPSK modulation is considered. We assume the same
total transmit power of the system with and without NC.

Fig. 2 presents the throughput performance of the system with and
without network coding, for $4$, $5$, and $6$ users. The ``X users
with NC" curves are generated by combining (\ref{exact N user error
probability 1}) and (\ref{throughput}), and the ``X users without
NC" curves are generated by combining (\ref{Sep no nc}) and
(\ref{throughput_no nc}). The NC matrices are given by (\ref{N user
matrix}). From the figure, it can be seen that compared to the
system without NC, NC improves the throughput about $0.21$, $0.22$,
$0.23$ for $4$, $5$, $6$ users respectively. This is predicted by
(\ref{B 3}), which indicates that as the number of the users
increases, the absolute value of the throughput improvement
increases. Moreover, for $4$ users, it can be seen that network
coding improves the throughput about $14 \%$ in the high SNR region,
while the improvement is about $11\%$ for $5$ users and about $9\%$
for $6$ users. For a practical communication service, such as video
conference in which each user may want to have a discussion with the
other users, the number of users is usually limited, e.g., 3 or 4
people when using Damaka$^{\text{TM}}$ \cite{video_conference}.
Furthermore, it is natural that the error performance gain will
decrease as the number of users increases for any MAC strategy. This
indicates that our proposed scheme can improve the throughput of the
system much in widely used scenarios.

In Fig. 3, the simulated and tight upper bounds of the SEP performance of the system with network
coding are compared for different number of users. The ``numerical" curves are generated by (\ref{SER
system_1}) and the NC matrices are given by (\ref{N user matrix}). We can see that the numerical SEP
curves accurately predict the simulation ones. From the figure, the SEP performance is slightly higher
as the number of users increases. It means that the interference between the users is small, which is
caused by the network coding at the relay. Thus when the number of the users is large, NC is still
efficient in improving the throughput with not so much impact on the system error performance.

In Fig. 4, the SEP performance of the system is compared with different network coding matrix, for $4$
users, where the matrices $1$, $2$, $3$ are, respectively,
\[
\left[ {\begin{array}{*{20}c}
   0 & 0 & 1 & 1  \\
   0 & 1 & 0 & 1  \\
   1 & 0 & 0 & 1  \\
\end{array}} \right],
\left[ {\begin{array}{*{20}c}
   0 & 0 & 1 & 1  \\
   0 & 1 & 1 & 0  \\
   1 & 1 & 0 & 0  \\
\end{array}} \right],
\left[ {\begin{array}{*{20}c}
   0 & 1 & 0 & 1  \\
   1 & 0 & 1 & 0  \\
   0 & 0 & 1 & 1  \\
\end{array}} \right].
\]It can be seen that our designed matrix can improve the SEP performance compared to an
ad-hoc coding matrix. However, compared to the system without NC, our proposed scheme achieves a
slightly poor SEP performance. The above simulations assume the same channel conditions of the source
to relay channel and the corresponding relay to source channel. However, in Fig. 5, we will show that
our proposed scheme may achieve a better SEP performance in other channel conditions.

In Fig. 5 when the source to relay link is $20$dB better than the corresponding relay to source link,
it can be seen that our proposed scheme achieves about $2$dB SEP gain compared to the system without
NC. This indicates that our proposed scheme can improve both the throughput and the error performance
in the system where the source to relay link is better than the corresponding relay to source link,
such as one satellite assists the information exchange of several base stations.


\section{Conclusions}
We have investigated the design of binary linear NC for $N$-way relay networks, where $N$ end nodes
exchange their information with the assistance of one $N$-way relay. NC matrix in GF(2) was proposed
to describe the linear NC process. The design criteria of the NC matrix, which improve the SEP
performance, were provided. Moreover, the closed form expressions and the upper bound of SEP of the
system were given. It can be seen that using linear NC, the throughput gain of the system is more than
$10\%$ for less than $6$ users. To improve the system performance further, we designed NC matrices for
arbitrary number of users, which minimized the bound of SEP.

\appendices
\section{Proof of Theorem 1}
First, we need to prove that it is the necessary condition. We suppose that all the network
sub-encoding matrices are full rank. Since \({\bf{F}}_i\) is full rank, the $N-1$ column vectors of \(
{\bf{F}}_i  \) form a $\left({N-1}\right)$-dimensional linear vector space in GF(2) and any other
column vector in GF(2) which contains $N-1$ elements is in this linear vector space. Then \({\bf{f}}_i
\) can be expressed as a kind of linear combining of the column vectors of \( {\bf{F}}_i \). That is
\begin{equation}
\label{Lemma1_2}
\begin{split}
{\bf{f}}_i  = \sum\limits_{k = 1,k \ne i}^N {^ \oplus \beta _k {\bf{f}}_k },
\end{split}
\end{equation}where \( {\beta _k } \) can only be $0$ or $1$. For
any user $j$, $j \ne i$, if \( {\beta _j } \) in (\ref{Lemma1_2}) is $0$, then \( {{\bf{F}}_j } \)
only has no more than $N-2$ linearly independent column vectors. Thus, \( {{\bf{F}}_j } \) is not full
rank. That is a contradiction of the hypothesis. So \( {\beta _k } \) in (\ref{Lemma1_2}) should be
$1$, for $ j \in \left[ {1,N} \right],j \ne i $. Thus we have proved that it is the necessary
condition.

Second, we need to prove that it is the sufficient condition. For \( j,j \ne i \), ${\bf{F}}_j$ can be
spread as
\begin{equation}
\label{Lemma1_3}
\begin{split}
 {\bf{F}}_j  &= \left[ {\begin{array}{*{20}c}
   {{\bf{f}}_1 }  &  \cdots  & {{\bf{f}}_{i - 1} } & {{\bf{f}}_i } & {{\bf{f}}_{i + 1} } &  \cdots  & {{\bf{f}}_{j - 1} } & {{\bf{f}}_{j + 1} } &  \cdots  & {{\bf{f}}_N }  \\
\end{array}} \right] _{(N-1) \times (N-1)}\\
  &= \left[ {\begin{array}{*{20}c}
   {{\bf{f}}_1 }  &  \cdots  & {{\bf{f}}_{i - 1} } & {\sum\limits_{k = 1,k \ne i}^N {^ \oplus {\bf{f}}_k } } & {{\bf{f}}_{i + 1} } &  \cdots  & {{\bf{f}}_{j - 1} } & {{\bf{f}}_{j + 1} } &  \cdots  & {{\bf{f}}_N }  \\
\end{array}} \right]_{(N-1) \times (N-1)} ,
\end{split}
\end{equation}where (\ref{Lemma1}) is used in the last step. Then
we add all the column vectors except the $i$th column vector onto the $i$th column in GF(2), which is
the elementary column operation on ${\bf{F}}_j$. Thus ${\bf{F}}_j$ can be expressed as
\begin{equation}
\label{Lemma1_3_2}
\begin{split}
  {\bf{F}}_j &\to \left[ {\begin{array}{*{20}c}
   {{\bf{f}}_1 }  &  \cdots  & {{\bf{f}}_{i - 1} } & {{\bf{f}}_j } & {{\bf{f}}_{i + 1} } &  \cdots  & {{\bf{f}}_{j - 1} } & {{\bf{f}}_{j + 1} } &  \cdots  & {{\bf{f}}_N }  \\
\end{array}} \right]_{(N-1) \times (N-1)} \\
  &\to \left[ {\begin{array}{*{20}c}
   {{\bf{f}}_1 }  &  \cdots  & {{\bf{f}}_{i - 1} } & {{\bf{f}}_{i + 1} } &  \cdots  & {{\bf{f}}_{j - 1} } & {{\bf{f}}_j } & {{\bf{f}}_{j + 1} } &  \cdots  & {{\bf{f}}_N }  \\
\end{array}} \right] _{(N-1) \times (N-1)}\\
  &= {\bf{F}}_i  .
\end{split}
\end{equation}

From (\ref{Lemma1_3_2}), it can be seen that ${\bf{F}}_j$ has the same rank of ${\bf{F}}_i$, which
means that \( {{\bf{F}}_j } \) is full rank if \( {{\bf{F}}_i } \) is full rank. Thus we have proved
that it is the sufficient condition and the theorem is proved.

\section{Proof of Theorem 2}
Using (\ref{the matrix of the coding in relay 2}), Eq. (\ref{decode 1_4}) can be rewritten as
\begin{equation}
\label{Proof1_3}
\begin{split}
{{{\bf{\hat x}}}_i} &= \left\{ {{\bf{F}}_i^{ - 1}\left\{ {{\bf{r}} \oplus {\bf{r}} \oplus {{{\bf{\tilde r}}}_i} \oplus \left\{ {\left( {{{\bf{f}}_i}{x_i}} \right)\bmod \left( 2 \right)} \right\}} \right\}} \right\}\bmod \left( 2 \right)\\
 &= \left\{ {{\bf{F}}_i^{ - 1}\left( {{\bf{r}} \oplus \left\{ {\left( {{{\bf{f}}_i}{x_i}} \right)\bmod \left( 2 \right)} \right\} \oplus \left( {{\bf{r}} \oplus {{{\bf{\tilde r}}}_i}} \right)} \right)} \right\}\bmod \left( 2 \right)\\
 &= \left\{ {{\bf{F}}_i^{ - 1}\left( {\underbrace {\left\{ {\left( {{\bf{F\tilde x}}} \right)\bmod \left( 2 \right)} \right\} \oplus \left\{ {\left( {{{\bf{f}}_i}{x_i}} \right)\bmod \left( 2 \right)} \right\}}_{ \buildrel \Delta \over = {I_1}} \oplus \left( {{\bf{r}} \oplus {{{\bf{\tilde r}}}_i}} \right)} \right)} \right\}\bmod \left( 2
 \right).
\end{split}
\end{equation}

The item $I_1$ in the above equation can be expressed as
\begin{equation}
\label{Proof1_4_reason 1}
\begin{split}
  I_1 &= \left[ {\begin{array}{*{20}c}
   {\sum\limits_{k = 1}^{k = N} {^\oplus f_{1 ,k } \tilde x_k } }  \\
   {\sum\limits_{k = 1}^{k = N} {^\oplus f_{2 ,k } \tilde x_k } }  \\
   {\cdots }  \\
   {\sum\limits_{k = 1}^{k = N} {^\oplus f_{N - 1 ,k } \tilde x_k } }  \\
\end{array}} \right] _{(N-1) \times 1} \oplus \left[ {\begin{array}{*{20}c}
   {f_{1 ,i } x_i }  \\
   {f_{2 ,i } x_i }  \\
   {\cdots}  \\
   {f_{N - 1 ,i } x_i }  \\
\end{array}} \right] _{(N-1) \times 1},
\end{split}
\end{equation}and the above expression can be rewritten
as
\begin{equation}
\label{Proof1_4_reason}
\begin{split}
   I_1&= \left[ {\begin{array}{*{20}c}
   {\sum\limits_{k = 1,k \ne i}^{k = N} {^\oplus f_{1 ,k } \tilde x_k } }  \\
   {\sum\limits_{k = 1,k \ne i}^{k = N} {^\oplus f_{2 ,k } \tilde x_k } }  \\
   {\cdots}  \\
   {\sum\limits_{k = 1,k \ne i}^{k = N} {^\oplus f_{N - 1 ,k } \tilde x_k } }  \\
\end{array}} \right]_{(N-1) \times 1} \oplus \left[ {\begin{array}{*{20}c}
   {f_{r_1 ,x_i } {\left( {{x_i} \oplus {{\tilde x}_i}} \right)} }  \\
   {f_{r_2 ,x_i } {\left( {{x_i} \oplus {{\tilde x}_i}} \right)} }  \\
   {\cdots}  \\
   {f_{r_{N - 1} ,x_i } {\left( {{x_i} \oplus {{\tilde x}_i}} \right)} }  \\
\end{array}} \right]_{(N-1) \times 1} \\
  &= \left\{ {\left( {{{\bf{F}}_i}{{{\bf{\tilde x}}}_i}} \right)\bmod \left( 2 \right)} \right\} \oplus \left\{ {\left( {{{\bf{f}}_i}\left( {{x_i} \oplus {{\tilde x}_i}} \right)} \right)\bmod \left( 2 \right)} \right\},
\end{split}
\end{equation}substituting (\ref{Proof1_4_reason}) into (\ref{Proof1_3}), (\ref{Proof1_3}) can be
written as
\begin{equation}
\label{Proof1_4}
\begin{split}
 {{{\bf{\hat x}}}_i} = \left\{ {{\bf{F}}_i^{ - 1}\left\{ {\left\{ {\left( {{{\bf{F}}_i}{{{\bf{\tilde x}}}_i}} \right)\bmod \left( 2 \right)} \right\} \oplus \left\{ {\left( {{{\bf{f}}_i}\left( {{x_i} \oplus {{\tilde x}_i}} \right)} \right)\bmod \left( 2 \right)} \right\} \oplus \left( {{\bf{r}} \oplus {{{\bf{\tilde r}}}_i}} \right)} \right\}} \right\}\bmod \left( 2 \right) .
\end{split}
\end{equation}

Since GF(2) is Galois field, distributive law of multiplication
exists. Using distributive law of multiplication, (\ref{Proof1_4})
equals to
\begin{equation}
\label{Proof1_4_1}
\begin{split}
{{{\bf{\hat x}}}_i} &= \left\{ {\left( {{\bf{F}}_i^{ - 1}{{\bf{F}}_i}{{{\bf{\tilde x}}}_i}} \right)\bmod \left( 2 \right)} \right\} \oplus \left\{ {\left( {{\bf{F}}_i^{ - 1}{{\bf{f}}_i}\left( {{x_i} \oplus {{\tilde x}_i}} \right)} \right)\bmod \left( 2 \right)} \right\} \oplus \left\{ {\left( {{\bf{F}}_i^{ - 1}\left( {{\bf{r}} \oplus {{{\bf{\tilde r}}}_i}} \right)} \right)\bmod \left( 2 \right)} \right\}\\
 &= {{{\bf{\tilde x}}}_i} \oplus \left\{ {\left( {{\bf{F}}_i^{ - 1}{{\bf{f}}_i}\left( {{x_i} \oplus {{\tilde x}_i}} \right)} \right)\bmod \left( 2 \right)} \right\} \oplus \left\{ {\left( {{\bf{F}}_i^{ - 1}\left( {{\bf{r}} \oplus {{{\bf{\tilde r}}}_i}} \right)} \right)\bmod \left( 2 \right)} \right\}
 .
\end{split}
\end{equation}

Using Theorem 1, since ${{\bf{F}}_i}$ should be full rank, we have
\begin{equation}
\label{Lemma1_3}
\begin{split}
   {{\bf{f}}_i} &
     = \sum\limits_{k = 1,k \ne i}^N {^\oplus {\bf{f}}_k }   \\
    &
 = \left( {\left[ {\begin{array}{*{20}{c}}
{{{\bf{f}}_1}}&{...}&{{{\bf{f}}_{i - 1}}}&{{{\bf{f}}_{i + 1}}}&{...}&{{{\bf{f}}_N}}
\end{array}} \right]{{\bf{1}}_{N - 1 \times 1}}} \right)\bmod \left( 2 \right)\\
 &= \left( {{{\bf{F}}_i}{{\bf{1}}_{N - 1 \times 1}}} \right)\bmod \left( 2 \right)
.
\end{split}
\end{equation}

Based on (\ref{Lemma1_3}), Eq. (\ref{Proof1_4_1}) can be expressed
as
\begin{equation}
\label{X_i and X'_i}
\begin{split}
{{{\bf{\hat x}}}_i} = {{{\bf{\tilde x}}}_i} \oplus \left( {\left( {{x_i} \oplus {{\tilde x}_i}}
\right){{\bf{1}}_{N - 1 \times 1}}} \right) \oplus \left\{ {\left( {{\bf{F}}_i^{ - 1}\left( {{\bf{r}}
\oplus {{{\bf{\tilde r}}}_i}} \right)} \right)\bmod \left( 2 \right)} \right\},
\end{split}
\end{equation}where ${\bf{F}}_i ^{ -
1}  {\bf{F}}_i ={\bf{I}}$ is used. Thus the theorem is proved.

\section{Proof of Lemma 3}
Mathematical induction is proposed to prove the theorem. First,
consider the addition of two numbers in GF(2),
\begin{equation}
\label{proof_XOR_1}
\begin{split}
a \oplus b = a + b - 2ab.
  \end{split}
\end{equation}

Then we assume that the addition of $Q$ numbers in GF(2) can be
written as
\begin{equation}
\label{proof_XOR_2}
\begin{split}
\sum\limits_{q = 1}^Q {^ \oplus  a_q }  = \sum\limits_{q = 1}^Q {\left( { - 2} \right)^{q - 1}
\sum\limits_{{1 \le {p_1} < {p_2} <  \cdots  < {p_q} \le Q} }^{} {\prod\limits_{j = 1}^q {a_{p_j } } }
}.
  \end{split}
\end{equation}

Using (\ref{proof_XOR_1}), the addition of $Q+1$ numbers in GF(2)
can be expressed as
\begin{equation}
\label{proof_XOR_3}
\begin{split}
 \sum\limits_{q = 1}^{Q + 1} {^ \oplus  a_q }  &= a_{Q + 1}  \oplus \sum\limits_{q = 1}^Q {^ \oplus  a_q }  \\
  &= a_{Q + 1}  + \sum\limits_{q = 1}^Q {^ \oplus  a_q }  - 2a_{Q + 1} \sum\limits_{q = 1}^Q {^ \oplus  a_q
  },
  \end{split}
\end{equation}and taking (\ref{proof_XOR_2}) into the above expression,
(\ref{proof_XOR_3}) can be rewritten as
\begin{equation}
\label{proof_XOR_4}
\begin{split}
 \sum\limits_{q = 1}^{Q + 1} {^ \oplus  a_q }  &= a_{Q + 1}  + \sum\limits_{q = 1}^Q {\left( { - 2} \right)^{q - 1} \sum\limits_{q = 1}^Q {\left( { - 2} \right)^{q - 1} \sum\limits_{{1 \le {p_1} < {p_2} <  \cdots  < {p_q} \le Q} }^{} {\prod\limits_{j = 1}^q {a_{p_j } } } } }  \\
  &\quad - 2a_{Q + 1} \sum\limits_{q = 1}^Q {\left( { - 2} \right)^{q - 1} \sum\limits_{{1 \le {p_1} < {p_2} <  \cdots  < {p_q} \le Q}}^{ } {\prod\limits_{j = 1}^q {a_{p_j } } } }
  .
   \end{split}
\end{equation}

Separate the $q=1$ term from the second item of (\ref{proof_XOR_4})
and separate the $q=Q+1$ term from the last item of
(\ref{proof_XOR_4}), we have
\begin{equation}
\label{proof_XOR_5}
\begin{split}
  \sum\limits_{q = 1}^{Q + 1} {^ \oplus  a_q } &= a_{Q + 1}  + \sum\limits_{p_1 }^Q {a_{p_1 } }  + \sum\limits_{q = 2}^Q {\left( { - 2} \right)^{q - 1} \sum\limits_{{1 \le {p_1} < {p_2} <  \cdots  < {p_q} \le Q}}^{ } {\prod\limits_{j = 1}^q {a_{p_j } } } }  \\
 &\quad +\sum\limits_{q = 2}^Q {\left( { - 2} \right)^{q - 1} \sum\limits_{{1 \le {p_1} < {p_2} <  \cdots  < {p_q} \le Q+1}}^{ } {\prod\limits_{j = 1}^q {a_{p_j } } } }  + \left( { - 2} \right)^Q \prod\limits_{p_1  = 1}^{Q + 1} {a_{p_1 }
 }.\\
   \end{split}
\end{equation}

Next, combining the first item and second item of
(\ref{proof_XOR_5}), and combining the third item and forth item of
(\ref{proof_XOR_5}), the above expression can be reformulated as
\begin{equation}
\label{proof_XOR_5}
\begin{split}
 \sum\limits_{q = 1}^{Q + 1} {^ \oplus  a_q } &= \sum\limits_{p_1 }^{Q + 1} {a_{p_1 } }  + \sum\limits_{q = 2}^{Q +
1} {\left( { - 2} \right)^{q - 1} \sum\limits_{{1 \le {p_1} < {p_2} <  \cdots  < {p_q} \le Q+1}}^{}
{\prod\limits_{j = 1}^q {a_{p_j } } } }  + \left( { - 2} \right)^Q \prod\limits_{p_1  = 1}^{Q + 1}
{a_{p_1 } }\\
 &= \sum\limits_{q = 1}^{Q + 1} {\left( { - 2} \right)^{q - 1} \sum\limits_{{1 \le {p_1} < {p_2} <  \cdots  < {p_q} \le Q+1}}^{ } {\prod\limits_{j = 1}^q {a_{p_j } } }
 }.
   \end{split}
\end{equation}

From the above derivation, (\ref{GF(2) add lemma}) is convenient for the addition of two numbers in
GF(2). Moreover, based on the addition of $Q$ numbers in GF(2), (\ref{GF(2) add lemma}) is also
convenient for the addition of $Q+1$ numbers in GF(2). Thus, using mathematical induction, for
arbitrary numbers, the addition in GF(2) can be expressed as (\ref{GF(2) add lemma}) and Lemma 3 is
proved.

\section{Proof of Lemma 4}
Eq. (\ref{GF(2) add}) is obvious and we focus on (\ref{GF(2) multiply}). We define that
$[{\bf{A}}]_{i,j}=a_{i,j}$ and $[{\bf{B}}]_i=b_i$. Then the multiplication between matrix ${\bf{A}}$
and ${\bf{B}}$ in GF(2) can be expressed as
\begin{equation}
\label{A mulitply B in GF(2)}
\begin{split}
\left( {{\bf{AB}}} \right)\bmod \left( 2 \right) = \left[ {\begin{array}{*{20}c}
   {\sum\limits_{i = 1}^M {^ \oplus  a_{1,i} b_i } }  \\
   {\sum\limits_{i = 1}^M {^ \oplus  a_{2,i} b_i } }  \\
    \vdots   \\
   {\sum\limits_{i = 1}^M {^ \oplus  a_{L,i} b_i } }  \\
\end{array}} \right]_{M \times 1}  \le \left[ {\begin{array}{*{20}c}
   {\sum\limits_{i = 1}^M {a_{1,i} b_i } }  \\
   {\sum\limits_{i = 1}^M {a_{2,i} b_i } }  \\
    \vdots   \\
   {\sum\limits_{i = 1}^M {a_{L,i} b_i } }  \\
\end{array}} \right]_{M \times 1} = {\bf{AB}}.
\end{split}
\end{equation}where $a \oplus b \le a + b$ is used in the inequality.

Thus we have (\ref{GF(2) multiply}) and the lemma is proved.

\section{Proof of Lemma 5}
For $N=3$, from (\ref{SER system}), the error probability of the
system in this situation can be calculated as
\begin{equation}
\label{3 user error probability}
\begin{split}
 3P_e  & \le 4\left( {E\left[ {{{x_1} \oplus {{\tilde x}_1}}} \right] + E\left[ { {{x_2} \oplus {{\tilde x}_2}}  } \right] + E\left[ {  {{x_3} \oplus {{\tilde x}_3}}  } \right]}
 \right)\\
  &\quad + \left( {f_{2 ,2 }  + f_{2 ,3 } } \right)E\left[ { {{r_1} \oplus {{\tilde
r}_{1,1}}} } \right]
 + \left( {f_{1 ,2 }  + f_{1 ,3 } } \right)E\left[ {  {{r_2} \oplus {{\tilde
r}_{1,2}}} } \right]
 \\
  &\quad + \left( {f_{2 ,1 }  + f_{2 ,3 } } \right)E\left[ {  {{r_1} \oplus {{\tilde
r}_{2,1}}} } \right]
 + \left( {f_{1 ,1 }  + f_{1 ,3 } } \right)E\left[ { {{r_2} \oplus {{\tilde
r}_{2,2}}}  } \right]
\\
  &\quad + \left( {f_{2 ,1 }  + f_{2 ,2 } } \right)E\left[ { {{r_1} \oplus {{\tilde
r}_{3,1}}}  } \right]
 + \left( {f_{1 ,1 }  + f_{1 ,2 } } \right)E\left[ {  {{r_2} \oplus {{\tilde
r}_{3,2}}} } \right]
 \\
  &= 4\left( {\left( {{x_1} \oplus {{\tilde x}_1}} \right) + \left( {{x_2} \oplus {{\tilde x}_2}} \right) + \left( {{x_3} \oplus {{\tilde x}_3}} \right)}
  \right)\\
  &\quad + \left( {E\left[ {  {{r_2} \oplus {{\tilde
r}_{2,2}}}  } \right]
 + E\left[ { {{r_2} \oplus {{\tilde
r}_{3,2}}}  } \right]
} \right)f_{1 ,1 }  + \left( {E\left[ {  {{r_1} \oplus {{\tilde r}_{2,1}}}  } \right] +E\left[ {  {{r_1} \oplus {{\tilde r}_{3,1}}}  } \right]} \right)f_{2 ,1 }  \\
  &\quad + \left( {E\left[ { {{r_2} \oplus {{\tilde r}_{1,2}}}  } \right] + E\left[ { {{r_2} \oplus {{\tilde r}_{3,2}}}  } \right]} \right)f_{1 ,2 }  + \left( {E\left[ { {{r_1} \oplus {{\tilde r}_{1,1}}}  } \right] + E\left[ {  {{r_1} \oplus {{\tilde r}_{3,1}}}  } \right]} \right)f_{2 ,2 }\\
  &\quad + \left( {E\left[ {  {{r_2} \oplus {{\tilde r}_{1,2}}}  } \right] + E\left[ {  {{r_2} \oplus {{\tilde r}_{2,2}}}  } \right]} \right)f_{1 ,3 }  + \left( {E\left[ {  {{r_1} \oplus {{\tilde r}_{1,1}}}  } \right] + E\left[ {  {{r_1} \oplus {{\tilde r}_{2,1}}}  } \right]} \right)f_{2 ,3 }
  .
\end{split}
\end{equation}

To ensure that each user can obtain other two users' information, the column vector of the encoding
matrix \({\bf{F}}\) can only be \( \left[ {1,1} \right]^T \), \( \left[ {0,1} \right]^T \) and \(
\left[ {1,0} \right]^T \). Using the assumptions of the statistic channel conditions between the users
and the relay and the power that the relay uses to broadcast the detected information, we have
\begin{equation}
\label{3 user order 1}
\begin{split}
E\left[ { {{r_1} \oplus {{\tilde r}_{1,1}}}  } \right]&> E\left[ {  {{r_1} \oplus {{\tilde r}_{2,1}}}
} \right] >
E\left[ {  {{r_1} \oplus {{\tilde r}_{3,1}}} } \right] ,\\
E\left[ {  {{r_2} \oplus {{\tilde r}_{1,2}}}  } \right] &> E\left[ { {{r_2} \oplus {{\tilde r}_{2,2}}}
} \right] >
E\left[ {  {{r_2} \oplus {{\tilde r}_{3,2}}}  } \right],\\
E\left[ {  {{r_2} \oplus {{\tilde r}_{1,2}}}  } \right]&> E\left[ {  {{r_1} \oplus {{\tilde r}_{1,1}}}  } \right],\\
E\left[ {  {{r_2} \oplus {{\tilde r}_{2,2}}}  } \right]&> E\left[ {  {{r_1} \oplus {{\tilde r}_{2,1}}}  } \right],\\
E\left[ {  {{r_2} \oplus {{\tilde r}_{3,2}}}  } \right]&> E\left[ {  {{r_1} \oplus {{\tilde r}_{3,1}}}  } \right]. \\
\end{split}
\end{equation}

From the above relationship, it can be seen that
\begin{equation}
\label{3 user order 2}
\begin{split}
& E\left[ { {{r_1} \oplus {{\tilde r}_{2,1}}}  } \right]+E\left[ {  {{r_1} \oplus {{\tilde r}_{3,1}}}  } \right]\\
&< \min \left({ {E\left[ { {{r_2} \oplus {{\tilde r}_{2,2}}}  } \right] +E\left[ {  {{r_2} \oplus
{{\tilde r}_{3,2}}}  } \right]},{E\left[ {  {{r_1} \oplus {{\tilde r}_{1,1}}}
 } \right] +E\left[ { {{r_1} \oplus {{\tilde r}_{3,1}}}  } \right]} } \right),\\
& \max\left({ {E\left[ { {{r_2} \oplus {{\tilde r}_{2,2}}}  } \right] +E\left[ {  {{r_2} \oplus
{{\tilde r}_{3,2}}} } \right]},{E\left[ {  {{r_1} \oplus {{\tilde r}_{1,1}}}
 } \right] +E\left[ { {{r_1} \oplus {{\tilde r}_{3,1}}}  } \right]} } \right) ,\\
& < \min \left({ {E\left[ {  {{r_2} \oplus {{\tilde r}_{1,2}}}  } \right] +E\left[ {  {{r_2} \oplus
{{\tilde r}_{3,2}}} } \right]},{E\left[ {  {{r_2} \oplus {{\tilde r}_{1,2}}}  } \right] +E\left[ {
{{r_2} \oplus {{\tilde r}_{3,2}}}  } \right]},{E\left[ { {{r_1} \oplus {{\tilde r}_{1,1}}}  } \right]
+E\left[ { {{r_1} \oplus {{\tilde r}_{2,1}}}  } \right]} } \right).\\
\end{split}
\end{equation}

Thus, ${E\left[ { {{r_1} \oplus {{\tilde r}_{2,1}}}  } \right] +E\left[ { {{r_1} \oplus {{\tilde
r}_{3,1}}}  } \right]}$ is the smallest one and to minimize the SEP of the system, in (\ref{3 user
error probability}), the coefficient of ${E\left[ {  {{r_1} \oplus {{\tilde r}_{2,1}}}  } \right]
+E\left[ { {{r_1} \oplus {{\tilde r}_{3,1}}}  } \right]}$ should be largest and then we have \( f_{2
,1 } = 1 \). Moveover, ${E\left[ {  {{r_2} \oplus {{\tilde r}_{2,2}}}  } \right] +E\left[ {  {{r_2}
\oplus {{\tilde r}_{3,2}}} } \right]}$ and ${E\left[ { {{r_1} \oplus {{\tilde r}_{1,1}}}  }
\right]+E\left[ { {{r_1} \oplus {{\tilde r}_{3,1}}}  } \right]} $ are smaller than the rest, which
result in \( f_{1 ,1 } =f_{2 ,2 } =1 \). Thus (\ref{3 user matrix}) is proved.

\bibliographystyle{IEEEtran}
\bibliography{IEEEabrv,myref}

 \newpage

\begin{figure}[!t]
\centering
\includegraphics[width=3in]{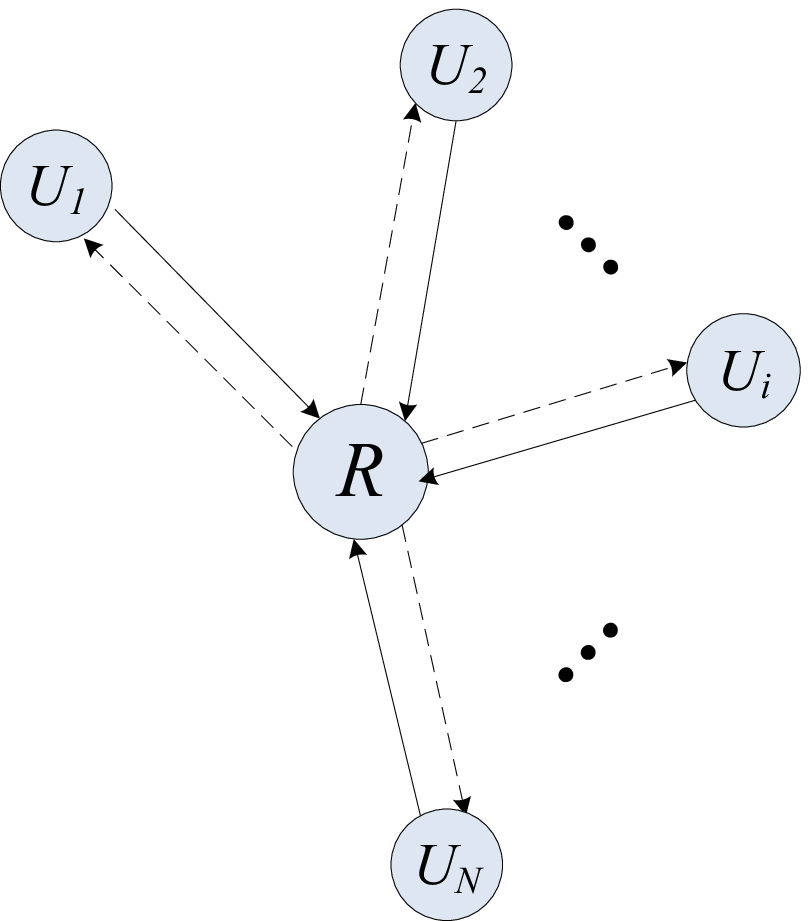}
\caption{System model.} \label{System model}
\end{figure}

\begin{figure}[!t]
\centering
\includegraphics[width=5in]{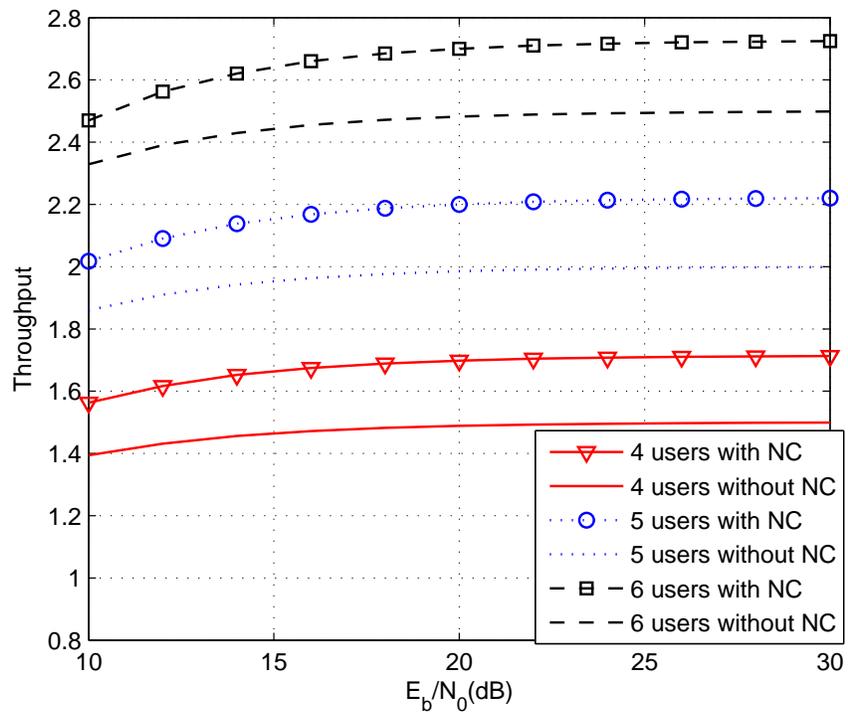}
\caption{Comparisons between the throughput of the system with and without network coding, for $4$,
$5$, and $6$ users.} \label{throughput_456.eps}
\end{figure}

\begin{figure}[!t]
\centering
\includegraphics[width=4.7in]{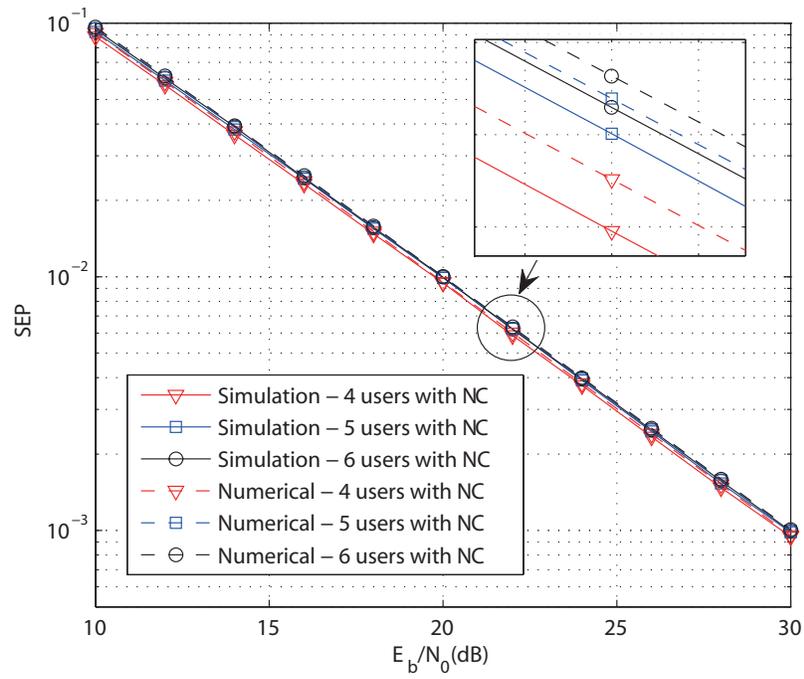}
\caption{The SEP of the system with network coding, for $4$, $5$, and $6$ users.}
\label{pingjun456_1.eps}
\end{figure}

\begin{figure}[!t]
\centering
\includegraphics[width=4.7in]{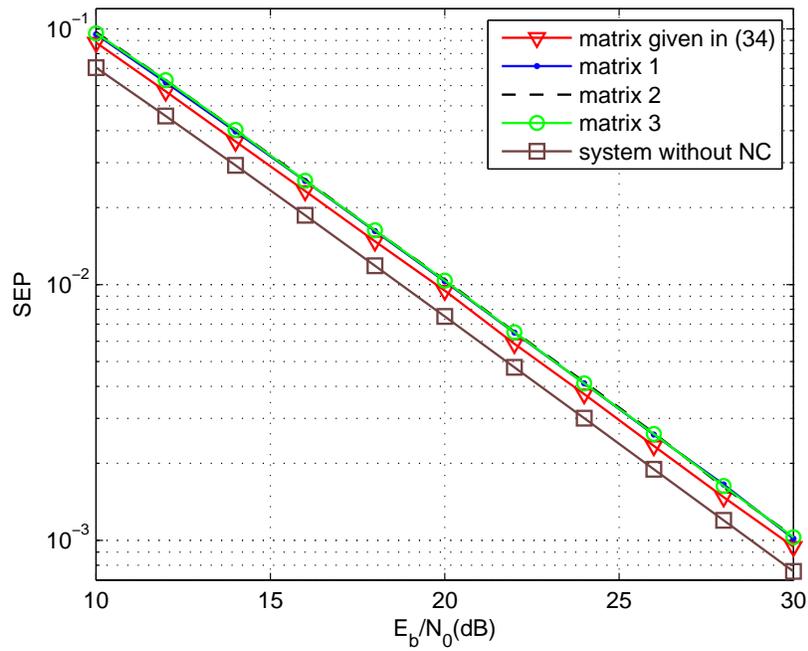}
\caption{The SEP of the system with different network coding matrix, for $4$ users.}
\label{yonghu_4_butongjuzhen.eps}
\end{figure}

\begin{figure}[!t]
\centering
\includegraphics[width=5in]{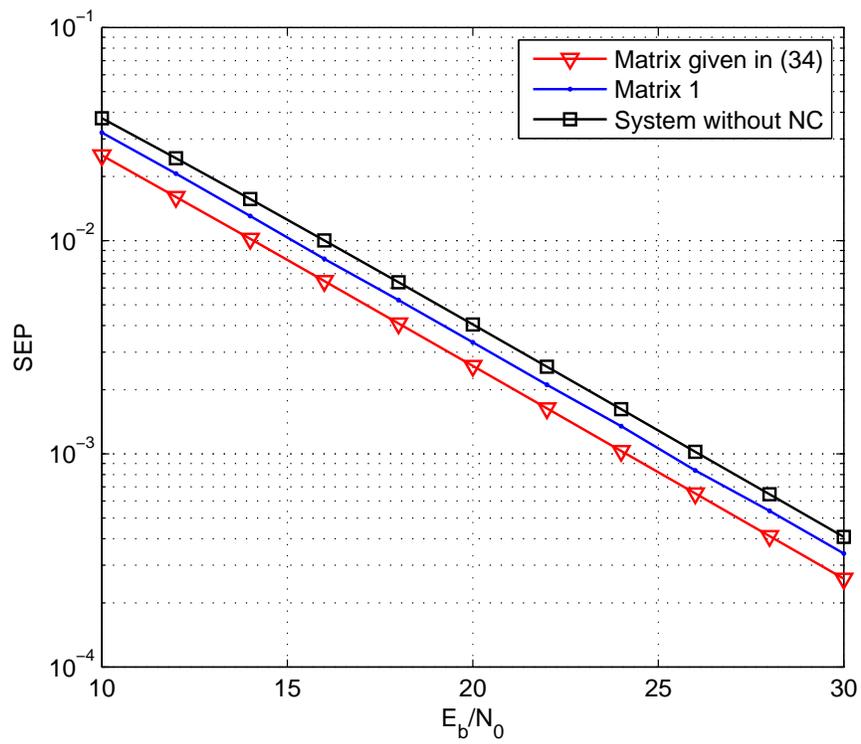}
\caption{The SEP of the system with different network coding matrix when the source to relay link is
$20$dB better than the corresponding relay to source link, for $4$ users.}
\label{yonghu_4_butongjuzhen_diertiao.eps}
\end{figure}

\end{document}